\documentclass[sigconf]{acmart}

\usepackage[utf8]{inputenc} 
\usepackage[T1]{fontenc}    
\usepackage{hyperref}       
\usepackage{url}            
\usepackage{booktabs}       
\usepackage{amsfonts}       
\usepackage{nicefrac}       
\usepackage{microtype}      
\usepackage[utf8]{inputenc} 
\usepackage[T1]{fontenc}    
\usepackage{hyperref}       
\usepackage{url}            
\usepackage{booktabs}       
\usepackage{amsfonts}       
\usepackage{nicefrac}       
\usepackage{microtype}      

\usepackage{amsthm, amsmath}
\usepackage{verbatim, float}
\usepackage{mathrsfs}
\usepackage{graphics}
\usepackage{subfig}
\usepackage[usenames,dvipsnames]{pstricks}
\usepackage{multirow}
\usepackage{url}
\usepackage{array}
\usepackage{tabu}
\usepackage{epsfig}
\usepackage{parallel}
\usepackage{graphicx}
\usepackage{diagbox}
\usepackage{mathtools}
\usepackage[absolute,overlay]{textpos}
\usepackage{wrapfig}
\usepackage{ulem}
\usepackage{stfloats}

\newcommand{\proj}{HIT\xspace}
\usepackage{enumitem}

\newcommand*\eg{\textit{e.g.}}
\newcommand*\ie{\textit{i.e.}}

\renewcommand{\emph}[1]{\textit{#1}}
\newcommand{\pan}[1]{\textcolor{blue}{(Pan: #1)}}

\newcommand{\yunyu}[1]{\textcolor{red}{(Yunyu: #1)}}

\theoremstyle{definition}
\newtheorem{theorem}{Theorem}[section]

\newtheorem{definition}[theorem]{Definition}

\theoremstyle{remark}

\usepackage[linesnumbered,vlined,ruled,commentsnumbered]{algorithm2e}

\AtBeginDocument{%
  \providecommand\BibTeX{{%
    \normalfont B\kern-0.5em{\scshape i\kern-0.25em b}\kern-0.8em\TeX}}}

\setcopyright{acmcopyright}
\copyrightyear{2022}
\acmYear{2022}
\acmDOI{10.1145/3437963.3441806}

\acmConference[WWW'22]{Proceedings of the ACM Web Conference 2022}{April 25--29, 2022}{Lyon, France.}
\acmBooktitle{Proceedings of the ACM Web Conference 2022 (WWW '22), April 25--29, 2022, Lyon, France}
\acmISBN{978-1-4503-9096-5/22/04}
\acmDOI{10.1145/XXXXXX.XXXXXX}

\begin{document}

\title{Neural Predicting Higher-order Patterns in Temporal Networks}


\author{Yunyu Liu}
\affiliation{%
  \institution{Purdue University}
  \city{West Lafayette, Indiana}
  \country{USA}
  }
\email{liu3154@purdue.edu}

\author{Jianzhu Ma}
\affiliation{%
  \institution{Peking University}
  \city{Beijing}
  \country{China}
  }
\email{majianzhu@pku.edu.cn}

\author{Pan Li}
\affiliation{%
  \institution{Purdue University}
  \city{West Lafayette, Indiana}
  \country{USA}
  }
\email{panli@purdue.edu}
\renewcommand{\shortauthors}{Y. Liu, J. Ma and P. Li}

\begin{abstract}
  Dynamic systems that consist of a set of interacting elements can be abstracted as temporal networks. Recently, higher-order patterns that involve multiple interacting nodes have been found crucial to indicate domain-specific laws of different temporal networks. 
  This posts us the challenge of designing more sophisticated hypergraph models for these higher-order patterns and the associated new learning algorithms. 
  Here, we propose the first model, named \emph{HIT}, for full-spectrum higher-order pattern prediction in temporal hypergraphs. 
  Particularly, we focus on predicting three types of common but important interaction patterns involving three interacting elements in temporal networks, which could be extended to even higher-order patterns. 
  \proj extracts the structural representation of a node triplet of interest on the temporal hypergraph and uses it to tell \emph{what} type of, \emph{when}, and \emph{why} the interaction expansion could happen in this triplet. 
  \proj could achieve significant improvement (averaged 20\% AUC gain to identify the interaction type, uniformly more accurate time estimation) compared to both heuristic and other neural-network-based baselines on 5 real-world large temporal hypergraphs. 
  Moreover, \proj provides a certain degree of interpretability by identifying the most discriminatory structural features on the temporal hypergraphs for predicting different higher-order patterns.
\end{abstract}





\maketitle


\setlength{\columnsep}{15pt}
\section{introduction}
Graphs provide a fundamental abstraction to study complex systems by viewing elements as nodes and their interactions as edges~\cite{newman2003structure}. 
Temporal graphs track the interactions over time, which allows for a more elaborate investigation of the underlying dynamics of network evolution~\cite{holme2012temporal,kovanen2011temporal}. 
Recently, numerous network models have been developed to learn the temporal network dynamics and further predict the future status of the complex system~\cite{snijders2010introduction,snijders2001statistical}. However, most of these models focus on predicting lower-order patterns in the graphs, \textit{e.g.}, edge prediction between a pair of nodes~\cite{liben2007link,rahman2016link,zhu2016scalable,ma2018graph}.

\begin{wrapfigure}{l}{0.25\textwidth}
    \centering
    \vspace{-5.2mm}
\includegraphics[trim={2.3cm 9.5cm 12.5cm 4.7cm},clip,width=0.25\textwidth]{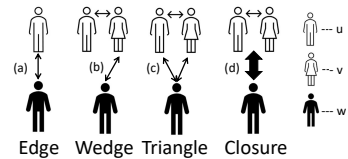}
  \vspace{-6.5mm}
\caption{\small{Incremental processes of the interaction expansion among $\{u,v,w\}$. The four patterns are formally defined in Sec.~\ref{sec:form}.  } }
\vspace{-3mm}
\label{fig:illu} \end{wrapfigure}

Higher-order interaction is a universal but more complicated phenomenon in many real-world systems~\cite{benson2016higher,lambiotte2019networks,li2017inhomogeneous,benson2021higher}. For instance, when we investigate the social circle of a family, one of the observed patterns we could study is that either one or two of the family members know the newly added member (Fig.~\ref{fig:illu}b-d), where the latter case obviously contains even more subtle different network patterns. It could be the case that both of the family members know the new one simultaneously in a common event (Fig.~\ref{fig:illu}d), or separately via different social events (Fig.~\ref{fig:illu}c). 

To fully characterize such higher-order patterns, it is necessary to use hypergraphs to model the complex system, in which hyperedges are elements involved in the same event~\cite{berge1984hypergraphs}. Recently, network scientists have confirmed the value of higher-order patterns in understanding network dynamics of multiple domains such as, social networks, transportation networks, and biological networks~\cite{benson2016higher,lambiotte2019networks,li2017inhomogeneous}, which posts us a great need to develop new algorithms to model the evolution of temporal hypergraphs. 
Successful models may be deployed in many multi-nodes-related applications, such as group recommendation~\cite{amer2009group} or monitoring synergic actions of multiple genes~\cite{perez2009understanding}. 



As shown in Fig.~\ref{fig:illu}, many new patterns might appear when we inspect how the interaction patterns of a group of nodes might change in a time-series manner and across different time scales. Identifying the dynamics that induce these subtle differences needs a more expressive network model. To the best of our knowledge, no previous works have successfully addressed this challenge. Heuristics proposed for edge prediction~\cite{liben2007link} have been generalized and confirmed for hyperedge prediction~\cite{benson2018simplicial,yoon2020much}, 
but those models 
marginally outperform random guessing and are far from ideal to be used in a practical system. 
Neural networks (NN) have more potential to encode complex structural information and recently achieved great success in various graph applications~\cite{hamilton2017representation,meng2018subgraph,jiang2019dynamic,sankar2020dysat,trivedi2019dyrep,kumar2019predicting,xu2020inductive,tgn_icml_grl2020}. However, 
previous works focused on either hyperedge prediction in static networks~\cite{rossi2019higher,yadati2020nhp,cotta2020unsupervised,alsentzer2020subgraph} or simple edge prediction in temporal networks but not for temporal higher-order patterns prediction. 
Moreover, none of the previous works (neither heuristic nor NN-based ones) are able to predict the entire spectrum of higher-order patterns (\textit{e.g.,} Fig.~\ref{fig:illu} b-d) let alone their temporal evolution.


\begin{figure*}[t]
    \centering
    \includegraphics[trim={0.7cm 11.7cm 0.4cm 3.5cm},clip,width=0.90\textwidth]{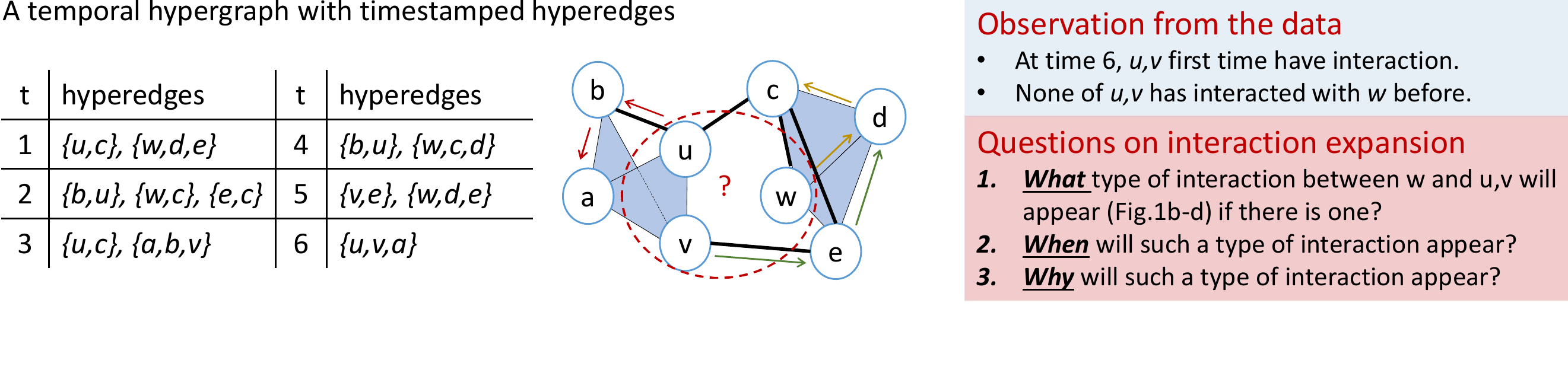}
        \vspace{-10mm}
    \caption{\small{Prediction problems defined on temporal hypergraphs. After the first time that $u, v$ interact, the questions focus on what type of, when and why such an interaction gets extended, \eg, to $w$.}}
    \label{fig:data-illu}
    \vspace{-1mm}
\end{figure*}

\begin{figure*}[t]
    \centering
    \includegraphics[trim={0.1cm 7.7cm 0.4cm 5.2cm},clip,width=0.90\textwidth]{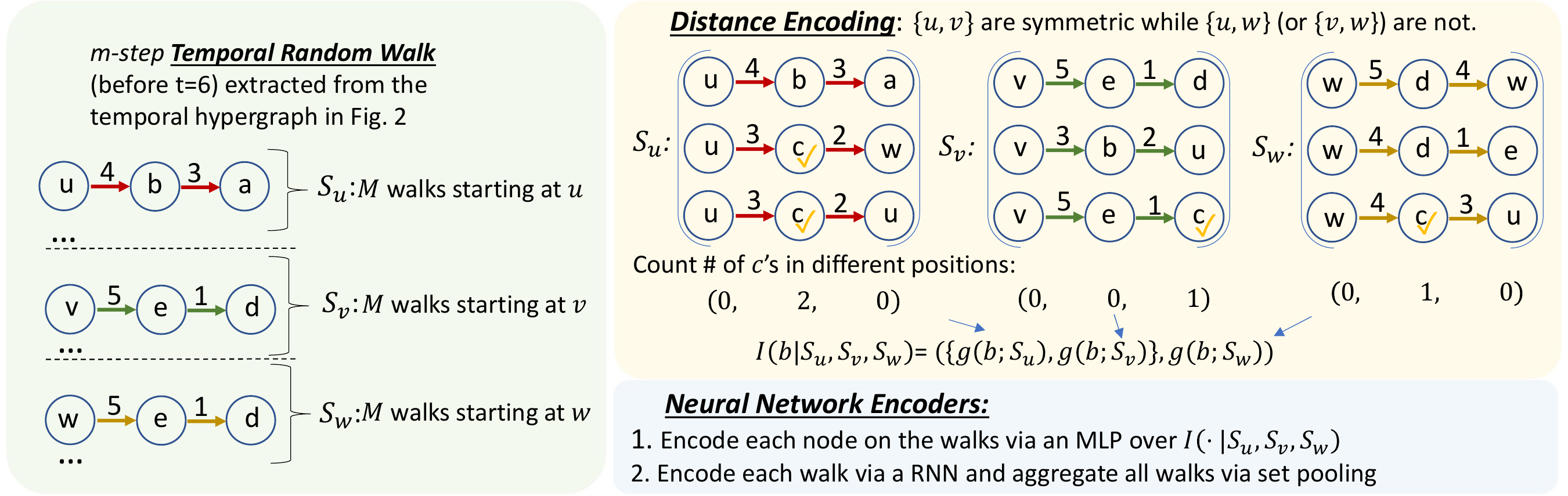}
        \vspace{-2mm}
    \caption{\small{Different components of \proj encoders: temporal random walks, distance encoding and neural network encoders.}}
    \label{fig:model}
    \vspace{-1mm}
\end{figure*}

Here, we propose a novel neural network model, termed \textbf{HI}gher-order \textbf{T}emporal structure predictor (\proj) to predict higher-order patterns in temporal hypergraphs. As proof of the concept, we focus on modeling the interaction patterns among three nodes. We intend to understand and model the \emph{interaction expansion} procedure on how a two-node interaction incorporates the third node via patterns such as a wedge, an (open) triangle or a (simplicial) closure~\cite{benson2018simplicial} (Fig.~\ref{fig:illu} b-d). The idea can be further generalized to learn more complicated patterns involving more nodes. 

The interaction expansion on the temporal hypergraph could be viewed as an incremental process as when the social circle expands, the number of interaction expansion rarely reduces. We target at predicting the \emph{first interaction expansion} among a triplet as it is of the most interest. This is far more challenging than conventional edge/hyperedge predictions in temporal networks, as no previous interactions within the same group of nodes can be leveraged. Fig.~\ref{fig:data-illu} uses a toy example to demonstrate the input data of our model and the three questions that we are about to answer. 

The key intuition of \proj is to encode the temporal network structure around the three nodes of interest using an NN to predict the patterns for their future interactions (Fig.~\ref{fig:model}).
\proj consists of three major components. First, \proj samples a few {\textit{temporal random walks}} (TRW) to efficiently collect the contextual structure. 
In order to model the network dynamics that lead to the first interaction expansion, no historical interactions within the same group of nodes should be referred to, 
and thus higher-order patterns must be captured in an inductive way via learning from other node triplets. Towards this goal, the second step of \proj is to represent the node triplets of interest without tracking the node identities by generalizing the recent {\textit{distance encoding}} (DE) technique~\cite{li2020distance}. 
Our new DE is built upon the counts that one node appears at a certain position according to the set of sampled TRWs while incorporating symmetry and asymmetry to fit the higher-order patterns. Finally, \proj imposes NNs to learn the representation of these DE-based TRWs. 
\proj adopts different decoders to solve different tasks regarding higher-order patterns. We conclude our contributions as follows, 
\begin{enumerate}[leftmargin=*]
    \item \proj is the first model to predict the full spectrum of higher-order patterns in interaction expansion from two nodes to a third node (Fig.~\ref{fig:illu} b-d). \proj achieves a significant boost in predicting accuracy (averaged 20\% gain in AUC) compared to both heuristic and NN-based baselines.    
    \item \proj is the first model designed for temporal hypergraphs. \proj provides elaborate time estimation about when interaction expands in certain patterns (uniformly better than baselines). 
    \item \proj can learn how different types of TRWs contribute to higher-order pattern prediction, which provides certain data insights. To the best of our knowledge, \proj is the first NN-based model that may characterize the discriminatory power of subtle dynamic structural features in temporal networks. We discover some insightful structures that yield a simplicial closure (Fig.\ref{fig:illu} d), a wedge (Fig.\ref{fig:illu} b), and no interaction expansion. 
\end{enumerate}

\vspace{-3mm}
\section{Related works}

Recently, higher-order patterns in networks have attracted much research interest because of their tremendous success in many real-world applications~\cite{milo2002network,alon2007network} including discovering data insights~\cite{benson2016higher,paranjape2017motifs,benson2018simplicial,lambiotte2019networks,do2020structural} and building scalable computing algorithms~\cite{yin2017local,paranjape2017motifs,fu2020local,veldt2020minimizing}.
Some works consider using higher-order motif closure to assist (lower-order) edge prediction~\cite{rossi2019higher, rossi2018hone, abuoda2019link, zhang2018link}. 
The works on predicting higher-order structures can be generally grouped into two categories, predicting multiple edges/subgraphs in graphs~\cite{lahiri2007structure,meng2018subgraph,nassar2020neighborhood,cotta2020unsupervised} and predicting hyperedges in hypergraphs~\cite{zhang2018beyond,benson2018simplicial,zhang2019hyper,yoon2020much,yadati2020nhp,alsentzer2020subgraph}. Subgraphs, \eg, cliques of nodes~\cite{benson2016higher}, could be used to describe higher-order patterns. However, they introduce a certain level of modeling confusion. For instance, current subgraph-based approaches were all found to fail to distinguish between triangles (Fig.~\ref{fig:illu} c) and closures (Fig.~\ref{fig:illu} d)~\cite{benson2018simplicial}. Hypergraph modeling could avoid such ambiguity but previous works mainly focus on predicting whether a hyperedge exists among a node set, which ignores other interaction patterns. Different from these two categories of works, our method aims to predict the entire spectrum of higher-order patterns. 


On the other hand, learning dynamics from temporal networks is also challenging.   
It is typically hard to incorporate simple heuristics such as commute time and PageRanks to encode graph structures to elaborate the complex patterns embedded in the temporal network~\cite{liben2007link,benson2018simplicial,rossi2019higher,yoon2020much,nassar2020neighborhood} due to their limited model expressivity. Therefore, more powerful NN-based models have been introduced to this domain but mostly are designed for normal temporal graphs instead of temporal hypergraphs. Many NN-based methods need to aggregate edges into network snapshots, which loses important temporal information
~\cite{meng2018subgraph,jiang2019dynamic,pareja2020evolvegcn,manessi2020dynamic,goyal2020dyngraph2vec,sankar2020dysat}. 
Other methods that learn node embeddings in the continuous time space may be able to predict higher-order temporal patterns by combining the node embeddings ~\cite{trivedi2019dyrep,kumar2019predicting,xu2020inductive,tgn_icml_grl2020}, but their performance was far from ideal due to their undesired generalization (See further explanations in Sec.~\ref{sec:encoder} and the experiments in Sec.~\ref{sec:exp}). Our work of encoding the contextual structure is inspired by a very recent work CAW-N~\cite{wang2020inductive}. However, their method is constructed for the task of edge prediction and not applicable to predicting higher-order patterns. 

\section{Problems Formulation}
\label{sec:form}

Temporal hypergraphs are hypergraphs with time information, in which each hyperedge corresponds to an interaction associated with a timestamp.  
\vspace{-0.5mm}
\begin{definition}
    [Temporal hypergraph] A \emph{temporal hypergraph} $\mathcal{E}$ can be viewed as a sequence of hyperedges with timestamps $\mathcal{E} = \{(e_1,t_1), (e_2,t_2), \cdots,  (e_N,t_N)\}, t_1 \leq t_2 \leq \cdots \leq t_N$,
    where $N$ is the number of hyperedges, and $e_i$, $t_i$ are the $i-$th hyperedge and its timestamp. Each hyperedge $e$ can be viewed as a set of nodes $\{v^{(e)}_{1}, v^{(e)}_{2}, \cdots, v^{(e)}_{|e|}\}$ where $|e|\geq 2$. 
\end{definition}
\vspace{-2mm}
\begin{definition}
    [Hyperedge covering] A node set $S$ is called covered by a hyperedge $e$ if $S\subseteq e$.
\end{definition}
\vspace{-2mm}
 For convenience, we use the following two concepts interchangeably in the rest of the paper: Nodes appear in an interaction and nodes are covered by a hyperedge.
 
 Edge prediction is to predict if two nodes connect to each other. 
 Here we generalize the concept of edge prediction to the higher-order pattern prediction for interaction expansion in hypergraphs.
 
\begin{definition} [Node triplets of interest] For three nodes $u, v, w$, any two of them have never interacted before timestamp $t$ and a hyperedge covers $u,v$ at time $t$ but not $w$, then we call $(\{u,v\},w,t)$ a node triplet of interest.
\end{definition}
\vspace{-2mm}
\begin{definition} [Interaction expansion] Given a node triplet of interest, say $(\{u,v\},w,t)$, interaction expansion refers to hyperedges connecting $u,v,w$ and they follow a pattern presented in Fig.~\ref{fig:illu} b-d. 
\end{definition}
\vspace{-1mm}
Compared with conventional edge prediction, interaction expansion prediction has three significant differences. First, instead of making a simple binary decision, the interaction expansion prediction becomes complicated as we incorporate multiple nodes, \eg, three patterns in interaction expansion from two nodes $\{u,v\}$ to another $w$ may happen (Fig.~\ref{fig:illu} b-d). 
Second, edge prediction in temporal graphs allows predicting the repetitive edges between two nodes at different timestamps. In contrast, interaction expansion is an incremental process where we need to predict how and when two nodes could extend their interaction for the first time. 
Although our model can also predict repetitive patterns among the same triplets, it is not the focus of this work.    
Third, this incremental process implies that any group of nodes may form a hyperedge at some time point (even after infinite time)~\cite{pishro2014introduction}. Therefore, 
it is more reasonable to set a time window $T_w$ and make predictions only within the time window. In particular, we propose three research problems based on these major differences.  

\textbf{Q1: For a triplet $(\{u,v\},w,t)$, what type of high-order interaction will be most likely to appear within $(t, t+T_w]$?}

To address this problem, we need to consider four possible interaction patterns as shown in Fig.~\ref{fig:illu}: 
1) \textbf{Edge}. There is no hyperedge covering $\{u,w\}$ or $\{v,w\}$ in $(t,t+T_w]$.
2) \textbf{Wedge}. Either $\{u,w\}$ or $\{v,w\}$ is covered by a hyperedge but not both. 
3) \textbf{Triangle}. $\{u,w\}$ and $\{v,w\}$ are covered by two hyperedges separately but there does not exist any hyperedge covering all three nodes $\{u,v,w\}$. 
4) \textbf{Closure}. All three nodes appear in at least one hyperedge. 

Note that these four patterns 
could be organized into a certain hierarchy. For instance, to form the Triangle pattern, the three nodes must first form a Wedge pattern within a shorter time. When there is a Closure pattern, hyperedges covering either $u,w$ or $v,w$ (corresponding to Wedge and Triangle) may still exist. This introduces some fundamental challenges to any computational models that try to predict these patterns. We further consider the next prediction problem to address this challenge. 

\textbf{Q2: Given a triplet $(\{u,v\},w,t)$ and a pattern in \{Wedge, Triangle, Closure\}, when will $u, v, w$ first form that pattern?}

Q2 asks for the exact timestamp of forming each pattern and thus avoids the potential overlap between these patterns. Q2 also shares some implications with Q1 so their predictions should be consistent. That is, if Q2 predicts that Wedge likely happens in a shorter time, Q1 should also assign a high probability to Wedge and vice versa. Answering Q2 may greatly benefit various time-sensitive applications. For instance, predicting the season when three dedicated travelers are grouped enables more accurate advertisement delivery for travel equipment. Predicting the exact timestamp when three genes mutually interact provides valuable information for molecular biologists to precisely control the process of cell development. 

\textbf{Q3: For a node triplet of interest, which structural features could be indicative to distinguish two patterns such as, Edge \textit{v.s.} Wedge,  Wedge \textit{v.s.} Triangle, Triangle \textit{v.s.} Closure?} 

Conventional heuristics generalized from edge prediction (\eg, Adamic–Adar Index~\cite{adamic2003friends}) have been proved as useful features to distinguish Triangle and Closure in static hypergraphs~\cite{benson2018simplicial}. We want to investigate whether \proj may leverage representation learning to discover more complicated and insightful structural features that distinguish higher-order patterns in temporal hypergraphs.

\section{Methodology}\label{sec:method}

In this section, we introduce the framework of \proj with an encoder-decoder structure. The structural encoders are designed to encode the contextual network structure around a triplet of interest. The problem-driven decoders are dedicated to providing answers to the three questions raised in Sec.\ref{sec:form}. 

\vspace{-1mm}
\subsection{Structural Encoders}
\label{sec:encoder}

To encode the structure of temporal hypergraphs, there are two fundamental challenges. 
First, temporal hyperedges that cover the same group of nodes may appear multiple times over time~\cite{lee2021thyme+}, which makes conventional structural encoders such as hypergraph GNNs~\cite{zhang2019hyper,feng2019hypergraph,yadati2020nhp} inscalable to model the entire contextual network. 
Second, as we always focus on the first time when interaction expanding appears, no historical behaviors we could rely on to make the predictions. Therefore, the encoder must be inductive and able to sufficiently represent the network context of the node triplets, which also helps to represent the fundamental dynamics governing network evolution. 
To address these challenges, we propose to use temporal random walks (TRWs) to extract the temporal structural features and pair them with asymmetric distance encoding (DE) to obtain the inductive representations, which later could be passed to a neural network to encode the DE-based TRWs. 

\textbf{Temporal Random Walks.} 
Given $(\{u,v\},w,t)$ a triplet of interest, our model \proj leverages TRWs to efficiently extracts the temporal structural features from its contextual network.  Specifically, for each node $z\in\{u,v,w\}$, \proj samples $M$ many $m$-step TRWs and group them into a set $S_z$ with only historical hyperedges sampled. The contextual structure can be represented by these three sets of TRWs $\{S_u, S_v, S_w\}$. The procedure is shown in Alg.~\ref{alg:TRW}.




 
        \begin{algorithm}
        \SetKwInOut{Input}{Input}\SetKwInOut{Output}{Output}
        \For{$j$ from $1$ to $M$, $W_j$ is the $j-$th walk}{
        
        Initialize $W_j$: $W_j[0]\leftarrow(z,t_0)$ \;
        \For{$i$ from $1$ to $m$, the $i-$th step}{
        $z_l, t_l \leftarrow W_j[i-1]$\;
        Sample one hyperedge $(e, t)\in \{(e,t) | z_l\in e, t < t_{l}\}$ with the probability $\propto$ $(|e| - 1) \exp(\alpha (t_{l}-t))$ \;
        Uniformly sample one $z' \in \{ e \setminus z_{l}\}$ \;
        $W_j[i]\leftarrow (z', t)$\;
        }
        }
        Return $S_z \triangleq \{W_j|1\leq j \leq M\}$\;
        \caption{\small{TRW Extraction ($\mathcal{E}$, $\alpha$, $M$, $m$, $z$, $t_0$)}}\label{alg:TRW}
        \end{algorithm}

\begin{table}
\vspace{-12mm}
\end{table}
 
Note that in a temporal hypergraph, more recent hyperedges with more nodes are often informative. Hence, we consider the $(|e| - 1)\exp(\alpha(t_l-t))$ as the probability to sample the hyperedge for some $\alpha\geq 0$, where $t$ is the timestamp of the corresponding hyperedge and $t_l$ is the timestamp of the previous hyperedge.
A larger $\alpha$ implies the tendency to sample more recent ones while $\alpha=0$ gives uniformly sampling. Each walk $W\in S_z$ can be viewed as a series of (node, time) pairs. 

Some previous work also applies similar TRW ideas to extract features from graphs~\cite{nguyen2018continuous}. However, their motivation of TRW is to mimic Deepwalk~\cite{perozzi2014deepwalk} for network embedding rather than improve scalability as ours. We have also tried to use hypergraph GNNs to work on the entire contextual network but have seen limited performance gain with much higher computational overhead.

\textbf{Asymmetric Distance Encoding.} To be well generalized across the time and the entire network, \proj cannot directly encode node identities in the sampled TRWs as one-hot vectors. 
However, if we simply remove node identities, the structural information underlying TRWs gets lost. Inspired by the recently proposed DE~\cite{li2020distance}, we compensate such information loss by generalizing the DE originally designed for static graphs to a new one that could fit hypergraphs. First, we encode a node $a$ based on the set $S_z$ of TRWs, denoted by $g(a;S_z)$, where the node $a$ is in $S_u\cup S_v\cup S_w$ and $z \in \{u,v,w\}$. 
The encoding $g(a;S_z)$, $z \in \{u,v,w\}$ is computed based on how many times the node $a$ has appeared on a certain position in $S_z$. In particular, we set $g(z; S_u)$ as a $(m+1)$-dim vector in which the $i$-th component can be computed as 
\begin{align}\label{eq:single-DE}
    g(a; S_z)[i] = |\{W|W\in S_z,\,a \text{ appears at }W[i]\}|,\,i=0,1,...,m
\end{align}
Another simpler encoding is to replace $g(a; S_z)$ with an indicator of an estimated shortest path distance (SPD) between $a$ and $z$, \ie, 
\begin{equation}
    \tilde{g}(a; S_z)\triangleq i,
    \text{where } i=\min\{j|\exists W\in S_z, a \text{ appears at }W[j]\}. 
    \label{eq:spd}
\end{equation} 
This encoding $\tilde{g}$ is better used to visualize our interpretation results for the question Q3, which will be further discussed in Sec.~\ref{sec:decoder}.  

Second, we aggregate the encodings $\{g(a;S_z)\}_{z\in\{u,v,w\}}$ into a vector representation $I(a|S_u,S_v,S_w)$ by considering the symmetry and asymmetry properties within the triplet $(\{u,v\},w,t)$.

It is important to properly model the symmetry (inductive bias) of a learning task, which has already yielded successful neural network architectures such as CNN for translating symmetry and GNN for permuting symmetry. 
In our setting, for a triplet $(\{u,v\},w,t)$, we have symmetry between $u$ and $v$ but they are asymmetric to $w$. 
The new DE function $I(a|S_u,S_v,S_w)$ is designed to maintain such symmetry for all the hyperedges, which could be written as 
\begin{align}\label{eq:DE}
I(a|S_u,S_v,S_w) \triangleq F(\{g(z; S_u),g(z; S_v)\},g(z; S_w))
\end{align}
Here $F$ is a mapping and the brace implies the permuting symmetry. As an instantiation, we set $I(a|S_u,S_v,S_w)=F_1((b + c)\oplus |b - c|)$, where $b = F_2(g(a; S_u)\oplus g(a; S_w))$, $c = F_2(g(a; S_v)\oplus g(a; S_w))$, $\oplus$ denotes the concatenation, both $F_1$ and $F_2$ are multi-layer perceptrons (MLP) and $|\cdot|$ is to compute element-wise absolute value. 

Then, we can represent each walk $W$ by replacing the node identity $a$ on $W$ with the encoding $I(a|S_u,S_v,S_w)$ based on
   \[((z_0,t_0), \cdots ,(z_m, t_m))\rightarrow W = ((I(z_0), t_0), \cdots, (I(z_m), t_m)),\]
where (and later) we use $I(a)$ to present $I(a|S_u,S_v,S_w)$ for brevity if causing no ambiguity. $(z_i,t_i)$ denotes the $i$-th (node, time) pair in the original walk. Next, we check two properties of our DE.

\emph{Generalizability:} 
As our problem is to predict the first interaction among a node triplet, no history among the same node triplet can be learnt from. Successful models must learn from the interactions of other triplets and further generalize such knowledge. We can actually show that our DE is generalizable (See proof in Appendix~\ref{sec:DE-explain}). Specifically, if two node triplets share the same historical contextual network structure, then any nodes that appear on TRWs sampled for these two node triplets hold the same DE distributions. In comparison, most previous approaches use node embeddings to compress the node-personalized history and update those embeddings over time. Our experiments in Sec.~\ref{sec:exp} show that the noise easily gets accumulated over time into node embeddings for these baselines, which yields undesired generalization performance. 

\emph{Expressivity:} We can also show our asymmetric DE is more expressive than the symmetric DE proposed in \cite{li2020distance}. Using our notation, the symmetric DE adopts a set-pooling $F'$,  $I_{\text{sym}}(a|S_u,S_v,S_w) \triangleq F'(g(z; S_u),g(z; S_v),g(z; S_w))$ which fits for the case when $u,v$ and $w$ are all symmetric. However, the three nodes are not symmetric in a node triplet $(\{u,v\},w,t)$. So, adopting $I_{\text{sym}}$ limits the model distinguishing  power, for which we show one example in Fig.~\ref{fig:asym-DE}. Note that one may consider setting the above $F'$ as the concatenating operation by even ignoring the symmetry between $u$ and $v$. Although our experiments show just marginal performance decay by using concatenation, it can hardly be generalized to predict even higher order of the patterns. This is because the dimension of concatenation linearly depends on the order of patterns. 


\begin{figure}
    \includegraphics[trim={1.8cm 10.9cm 7.8cm 4.3cm},clip,width=\columnwidth]{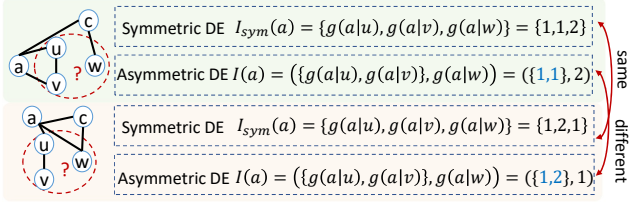}
    \vspace{-8mm}
    \caption{\small{Illustration of Symmetric DE and Asymmetric DE. For simplicity, we use static graphs instead of temporal hypergraphs to illustrate. Moreover, $g(a|z)$ here adopts the shortest path distance from $z$ to $a$ instead of the count of random-walk hits.}}
    \label{fig:asym-DE}
    \vspace{-4mm}
\end{figure}

\textbf{Neural-Network Encoding TRWs.} Now for each triplet, we collect a set of random walks in $S_u\cup S_v\cup S_w$, on which the nodes are encoded by DEs. Next, we employ neural networks to encode these walks into the final representation of this triplet. 
We adopt a RNN~\cite{rumelhart1986learning,hochreiter1997long} to encode each walk $W\in S_u\cup S_v\cup S_w$ by $\text{Enc}(W) = \text{RNN}(\{I(z_i)\oplus F_3(t_i)\}_{i=0}^m)$ where $F_3(\cdot)$ is a time encoder based on learnable Fourier features~\cite{kazemi2019time2vec,xu2019self} by following
\begin{equation*}
    F_3(t_i) \triangleq [\cos(\beta_1t_i)+\phi_1, \cos(\beta_2t_i)+\phi_2, \cdots, \cos(\beta_dt_i)+\phi_d],
\end{equation*}
where $\beta_j,\,\phi_j,\,j=\{1, 2, \cdots, d\}$ are learnable parameters. 
The learnable Fourier features increase the dimension of the scaled time value. Ensured by the Bochner's Theorem~\cite{bochner1992theorem}, it achieves more expressiveness while keeping invariant to time shift. 

We then apply a set pooling method AGG to calculate the encoding of nodes $u$, $v$, and $w$. 
\vspace{-1mm}
\begin{align} \label{eq:node-embedding}
     \psi(z) = & \text{AGG}(\{Enc(W)\}_{W\in S_z}),\, \text{for $z\in \{u,v,w\}$}.
\end{align}
We instantiate AGG as self-attention pooling and mean pooling. Both significantly outperform the baselines and self-attention pooling always achieves slightly higher performance than mean pooling. However, the mean pooling operation is more interpretable. Therefore, we choose mean pooling in our analysis for the question Q3 while self-attention pooling for Q1, Q2. Their specific forms are as follows: For a set of walks $S_z=\{W_1,W_2,...,W_M\}$,
\small{Self-Att($S_z$) = $\frac{1}{M}\sum_{i=1}^{M} \text{softmax}(\{\text{enc}(W_i)^T \Theta_1 \text{enc}(W_j)\}_{1\leq j\leq M}) \text{enc}(W_i) \Theta_2$}\normalsize, where $\Theta_1$, $\Theta_2$ are learnable parameter matrices, and 
\small{Mean($S_z$) = $\frac{1}{M}\sum_{i=1}^{M} \text{enc}(W_i)$}. \normalsize
 \vspace{-1mm}
\subsection{Problem-driven Decoders}\label{sec:decoder} 

Unlike the encoder, we need to design different decoders to answer three questions defined in Sec.\ref{sec:form}. To simplify the notations, we denote the node triplet of interest $(\{u,v\},w,t)$ as the letter $\tau$.

\textbf{Decoder for Q1.} Q1 is to predict the types of higher-order patterns of the interaction expansion and thus can be formalized as a standard multi-class classification problem. Similar to DE, we also consider the symmetric property of the problem when designing the decoder function and predict the probability distribution over the four different patterns as $\hat{Y}(\tau) \triangleq \text{softmax}(F_4((\psi(u) + \psi(v))\oplus  \psi(w)))$,
where $F_4$ is an MLP, $\hat{Y}(\tau)$ is a normalized 4-dim vector (for four classes). We train the model by minimizing the cross entropy loss between $\hat{Y}(\tau)$ and the ground truth pattern $Y(\tau)$.






 \textbf{Decoder for Q2.} In this task, we aim to model the distribution of the time when a certain interaction pattern appears after observing a node triplet of interest. We adopt a log-normal mixture distribution to model the potential multimodality of the distribution of the time. Given a pattern $p\in$ \{Wedge, Triangle, Closure\}, we assume the predicted timestamp $\hat{t}$ follows the distribution, $\log \hat{t}_{\tau,p} \sim \sum_{i=1}^k w_{p,i} \mathcal{N}(\mu_{p,i},\sigma_{p,i}^2)$, 
where the weight $w_p$, the mean $\mu_p$ and the standard deviation $\sigma_p$ are parameterized as follows. 
 \begin{align*}
    w_p = \text{softmax}(F_{w_p}(h)),\,\mu_p = F_{\mu_p}(h),\,\sigma_p^2 = \exp(F_{\sigma_p^2}(h)),
    \end{align*}
where $h = (\psi(u)+\psi(v))\oplus \psi(w)$, and $F_{w_p}$, $F_{\mu_p}$ and $F_{\sigma^2_p}$ are three MLPs with different parameters for different pattern $p$'s. 

To train the time predictor, we maximize the likelihood for each node triplet $\tau$. That is,  if the first time $\tau$ shows the pattern $p$ is $t+t_{\tau,p}$, we calculate the negative log-likelihood loss (NLL) as follows: \vspace{-1mm} 
\small{
\begin{equation}\label{eq:NLL}
NLL(\tau, p) = - \log \left[\sum_{i=1}^k w_{p,i}^{(l)} \frac{1}{\sqrt{2\pi\sigma_{p,i}^2}} \exp\left(-\frac{(\log t_{\tau,p}-\mu_{p,i})^2}{2\sigma_{p,i}^2}\right)\right].
 \end{equation}} \vspace{-1.5mm}
 \normalsize{}
\textbf{Decoder for Q3.} The decoder for Q3 is expected to hold certain interpretability. We aim to find the most discriminatory structural features for different higher-order patterns. 

There are two challenges: a) How to represent structural features in an interpretable way; b) How to characterize their importance to the predicted higher-order patterns. In what follows, we address challenge a) by categorizing TRWs according to DE properly. We address challenge b) by characterizing the linear impact of TRWs on the predicted higher-order patterns, which is more interpretable.

TRWs could be viewed as network motifs~\cite{kovanen2011temporal} or potentially informative temporal structures. TRWs whose nodes are paired with different DE distributions naturally characterize different types of structural features. So we categorize TRWs according to the DEs of the nodes. To obtain more visualizable results, we adopt the simpler DE $\tilde{g}$ (Eq.\eqref{eq:spd}), i.e., shortest path distance instead of $g$ (Eq.\eqref{eq:single-DE}). Note that $\tilde{g}$ can be obtained via a surjective mapping of the original $g$ and thus has less representation power. However, we observe \proj equipped with $\tilde{g}$ still significantly outperforms all baselines to identify higher-order patterns.

To impose linear impact, we adopt a simple linear logistic regression with DE-based TRWs as features. To identify two different patterns  $p_1,p_2$, each DE-based TRW, after neural encoding, will be tranformed in one score $C_W$. The value of $C_W$ essentially characterizes how likely the TRW $W$ may induce $p_1$ as opposed to $p_2$. Those TRWs with the largest and smallest $C_W$ are the most discriminative features. Specifically, we  minimize the following loss function

\small{
\begin{align}
    \sum_{\tau: Y(\tau)\in \{p_1,p_2\}}&\ell(x(\tau), Y(\tau)),\quad 
    \text{where}\; x(\tau)= \sum_{W\in S_u\cup S_v\cup S_w} C_W + b, \label{eq:TRW-score}
\end{align}}
\normalsize{and where $\ell(x, Y) = x *1_{Y=p_1} - \log(1+ \exp(x))$,  $Y(\tau)$ is the ground truth label of a node triplet $\tau$. We compute $C_W = B^T Enc(W)$ ($Enc(W)$ defined in Eq.\eqref{eq:node-embedding}), $B, b$ are also trainable weights and bias. 
Eq.\eqref{eq:TRW-score} can be achieved by setting AGG in Eq.\eqref{eq:node-embedding} as mean pooling plus a learnable bias. A larger $C_W$ refers to a TRW that is more indicative to $p_1$  as opposed to $p_2$ and vice versa.}

\begin{table*}[tp]
\begin{minipage}{0.58\textwidth}
\small
\begin{center}
\scalebox{0.9}{
\begin{tabular}{c c c c c c c c }
\toprule  & \footnotesize{\# Node} & \footnotesize{\# Edge}$^*$ & \footnotesize{Avg.}$^*$ & \footnotesize{Std.}$^*$& \footnotesize{\# Closure} & \footnotesize{\# Triangle} & \footnotesize{\#Wedge} \\
\midrule
tags-math-sx & 1.63K & 822K & 2.75 & 0.89 & 23.6K & 68.5K & 3.07M \\
tags-ask-ubuntu & 3.03K & 271K & 3.11 & 1.02 & 73.2K & 156K & 9.43M\\
congress-bills & 1.72K & 260K & 7.55 & 6.68 & 4.27M & 4.02M & 26.4M \\
DAWN & 2.56K & 2.27M & 2.59 & 1.22 & 45.8K & 62.3K & 3.13M \\
threads-ask-ubuntu & 125K & 192K & 2.30 & 0.63 & 7.44K & 23.9K & 3.99M \\
\bottomrule 
\end{tabular}
}
\end{center}

\end{minipage}
\hfill
\scalebox{0.95}{
\begin{minipage}{0.40\textwidth}
\includegraphics[trim={0.3cm 0.4cm 0.5cm 0.7cm},clip,width=0.49\textwidth]{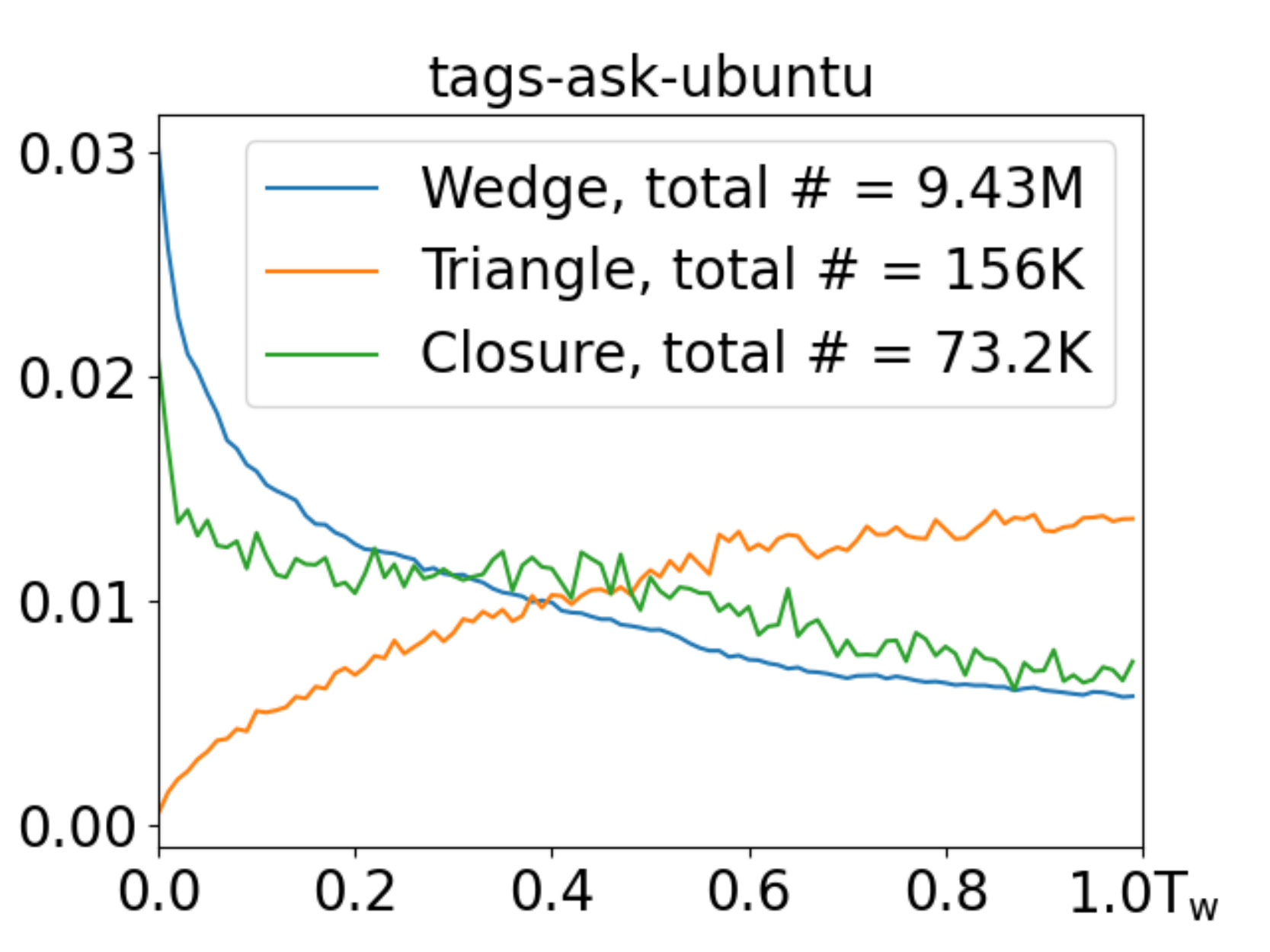}
\includegraphics[trim={0.3cm 0.4cm 0.5cm 0.7cm},clip,width=0.49\textwidth]{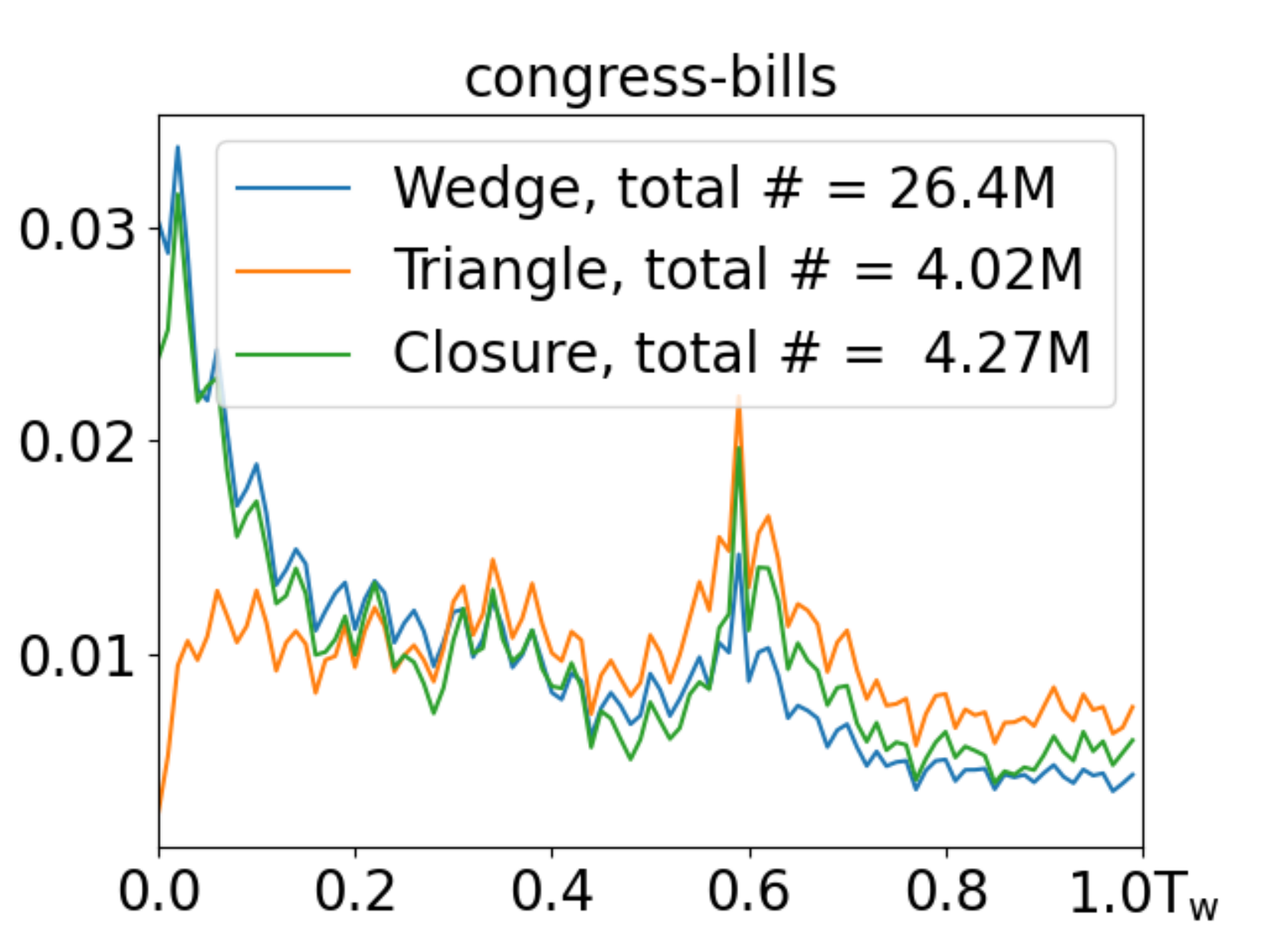}

\end{minipage}
}

\caption{\small{Basic statistics of the datasets. a) The numbers of different patterns in \{Wedge, Triangle, Closure\} within a time window $T_w = 0.1 * \text{the total time range}\,$ after initial interactions between the first two nodes get observed. * Edge is hyperedges. Ave. and Std. are averaged size of hyperedges and derivation. b) shows the distributions (P.D.F.) of the times used to expand such interactions to the third nodes.}} 
    \label{fig:data-stat}
    \vspace{-6mm}
\end{table*}

\begin{textblock*}{1cm}(1.5cm,2.98cm) 
   \textbf{(a)}
\end{textblock*}
\begin{textblock*}{1cm}(12.00cm,2.98cm) 
   \textbf{(b)}
\end{textblock*}

\section{Experiments}

\label{sec:exp}


Here, we conduct an extensive evaluation of \proj to answer the three questions raised in Sec.\ref{sec:form}.

\textbf{Datasets.} We use five real-world temporal hypergraph datasets collected at~\cite{benson2018simplicial}: Tags-math-sx, Tags-ask-ubuntu, Congress-bills, DAWN, Threads-ask-ubuntu. Their detailed description is left in Appendix~\ref{sec:data}. The basic statistics of these datasets are shown in Table~\ref{fig:data-stat}. Note that repetitive temporal hyperedges significantly enlarge the network. 
Since the complexity is mainly dominated by the number of hyperedges (from $192K$ to $2.27M$), scalable models are needed. Moreover, the higher-order patterns can be extremely unbalanced. We show the distributions of time used in interaction expanding for different patterns 
in Table~\ref{fig:data-stat} b). Triangles as expected take more time than Wedges while Closures unexpectedly share a similar tendency as Wedges.

\textbf{Data Preprocessing and Setup.} Ideally, in temporal hypergraphs, hyperedges continue appearing over time but the datasets only include a certain time range $T$. To avoid the effect of the boundary truncation, we focus on all the node triplets between $[0.4T, 0.9T]$. We use the hyperedges in $[0T, 0.4T]$ to set up the initial network. Node triplets in $[0.4T, 0.75T]$ are selected for model training, $[0.75T, 0.825T]$ for validating and $[0.825T, 0.9T]$ for test. We also set the time window defined in Q1 as $T_w=0.1T$, so the hyperedges in $[0.9T, T]$ essentially complete the interaction expansion for the node triplets in the test set. We enumerate all the Edges, Wedges, Triangles, and Closures as the data instances with ground truth labels and properly balance all classes (see more details in Appendix~\ref{sec:proc}). 
The performance of all methods is averaged across five-time independent experiments. 




\textbf{Baselines.} To evaluate our method, we compare \proj with ten baseline methods, including both heuristic methods and NN-based encoders. More detailed description is left in the Appendix~\ref{sec:base}.

\emph{Heuristic methods} are the metrics previously used for edge prediction~\cite{liben2007link} and recently got generalized for hyperedge prediction~\cite{benson2018simplicial,yoon2020much,nassar2020neighborhood}, including $3$-way Adamic Adar Index (\textbf{$3$-AA}), $3$-way Jaccard Coefficient (\textbf{$3$-JC}), $3$-way Preferential Attachment(\textbf{$3$-PA}), and the Arithmetic Mean of traditional \textbf{AA}, \textbf{JC} and \textbf{PA}.  
    Given a node triplet $(\{u,v\},w,t)$, we project hyperedges before $t$ into a static hypergraph, compute these metrics, and then impose a two-layer nonlinear neural network for prediction. 
 \emph{NN-based encoders} show great power in learning graph representations, but none of the previous models can be directly applied to hyperedge prediction in temporal hypergraphs let alone the full spectrum of higher-order patterns.
 We adopt a hypergraph representation learning model \textbf{NHP}~\cite{yadati2020nhp} which is proposed for hyperedge prediction in static hypergraphs. To adapt to temporal hypergraphs, we aggregate hyperedges into 10 hypergraph snapshots. We train an \textbf{NHP} over these snapshots to get temporal node embeddings and apply our decoders to make predictions. 
 We choose baselines \textbf{JODIE}~\cite{kumar2019predicting}, \textbf{TGAT}~\cite{xu2020inductive}, and \textbf{TGN}~\cite{tgn_icml_grl2020} as they are representatives models over temporal graphs. For a triplet $(\{u,v\},w,t)$, we use \textbf{JODIE}, \textbf{TGAT}, and \textbf{TGN} to generate node embeddings of $u,v,w$ and apply our decoders on these embeddings to make predictions. 


\begin{table*}[!ht]
\begin{minipage}{0.65\textwidth}
\resizebox{\textwidth}{!}{
\begin{tabular}{c c c c c c c c}
\toprule ~  & \small{tags-math-sx} & \small{tags-ask-ubuntu}  & \small{congress-bills} & \small{DAWN} & \small{threads-ask-ubuntu}\\
\midrule

3-AA & 61.86 $\pm$ 0.25 & 68.31 $\pm$ 0.10 & 68.20 $\pm$ 0.21 & 72.75 $\pm$ 1.36 & 63.86 $\pm$ 0.16\\

3-JC & 58.62 $\pm$ 0.79 & 65.24 $\pm$ 0.18 & 69.22 $\pm$ 0.50 & 70.67 $\pm$ 0.59 & 63.59 $\pm$ 0.79\\

3-PA & 58.63 $\pm$ 0.79 & 55.96 $\pm$ 2.08 & 50.69 $\pm$ 0.62 & 51.48 $\pm$ 1.36 & 63.49 $\pm$ 1.49\\

AA & 65.05 $\pm$ 0.52 & 69.74 $\pm$ 0.28 & 70.22 $\pm$ 0.73 & 75.71 $\pm$ 0.36 & 76.52 $\pm$ 0.36 \\


JC & 60.76 $\pm$ 0.34 & 67.07 $\pm$ 0.11 & 68.24 $\pm$ 0.37 & 70.48 $\pm$ 0.63 & 68.84 $\pm$ 1.77\\

PA & 56.53 $\pm$ 2.04 & 65.71 $\pm$ 1.53 & 53.48 $\pm$ 3.19 & 65.60 $\pm$ 3.14 & 70.03 $\pm$ 1.10\\


NHP & 50.78 $\pm$ 0.29 & 51.42 $\pm$ 0.23 & 49.48 $\pm$ 0.12 & 49.90 $\pm$ 0.07 & 52.04 $\pm$ 0.79 \\

TGAT & 63.27 $\pm$ 0.27 & 53.75 $\pm$ 0.16 & 56.68 $\pm$ 0.22 & 46.63 $\pm$ 0.16 & 81.92 $\pm$ 0.25\\

JODIE & 56.68 $\pm$ 0.12 & 61.84 $\pm$ 0.48 & 63.50 $\pm$ 0.24 & 52.66 $\pm$ 0.27 & 67.23 $\pm$ 0.58 \\

TGN & 55.79 $\pm$ 0.16 & 54.33 $\pm$ 0.67 & 66.63 $\pm$ 0.23 & 71.84 $\pm$ 0.31 & 80.13 $\pm$ 0.79\\



\bf{\proj} & \bf{74.07 $\pm$ 0.46} & \bf{78.83 $\pm$ 0.43} & \bf{79.83 $\pm$ 0.61} & \bf{78.92 $\pm$ 0.61} & \bf{84.22 $\pm$ 0.68} \\






\bottomrule 

\end{tabular}
}
\caption{\small{1-vs-1 AUC (mean$\pm$std) of higher-order pattern prediction (results for Q1): 1-vs-1 AUC is obtained by averaging AUCs of comparing the ground-truth class \text{v.s.} every other class in the multi-class classification. 
} 
}\label{tab:problem1_performance}
\end{minipage}
\hfill
\begin{minipage}{0.33\textwidth}
\vspace{-3mm}
\includegraphics[trim={0.3cm 0.5cm 0.5cm 0.7cm},clip,width=0.49\textwidth
]{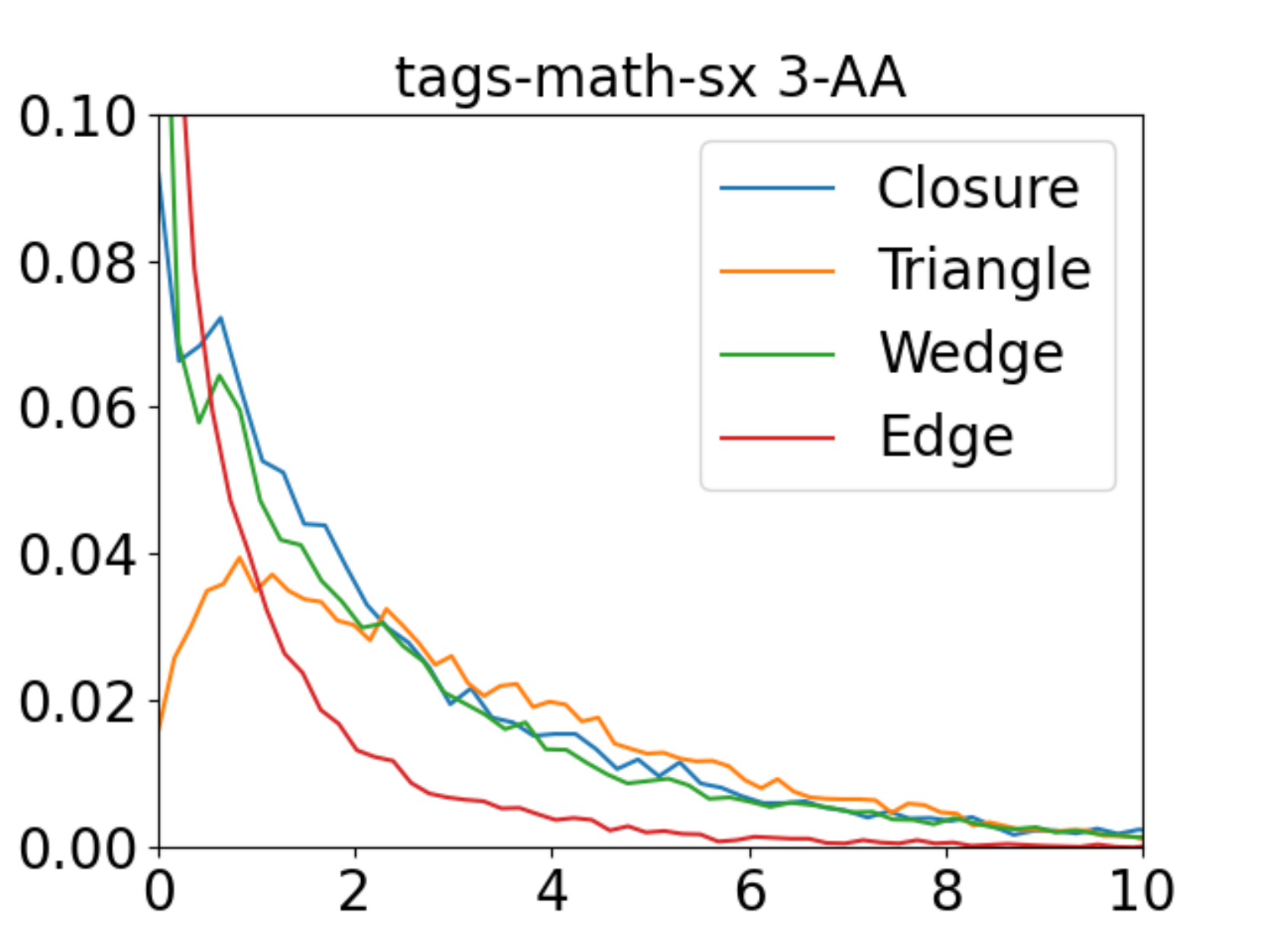}
\includegraphics[trim={0.3cm 0.5cm 0.5cm 0.7cm},clip,width=0.49\textwidth
]{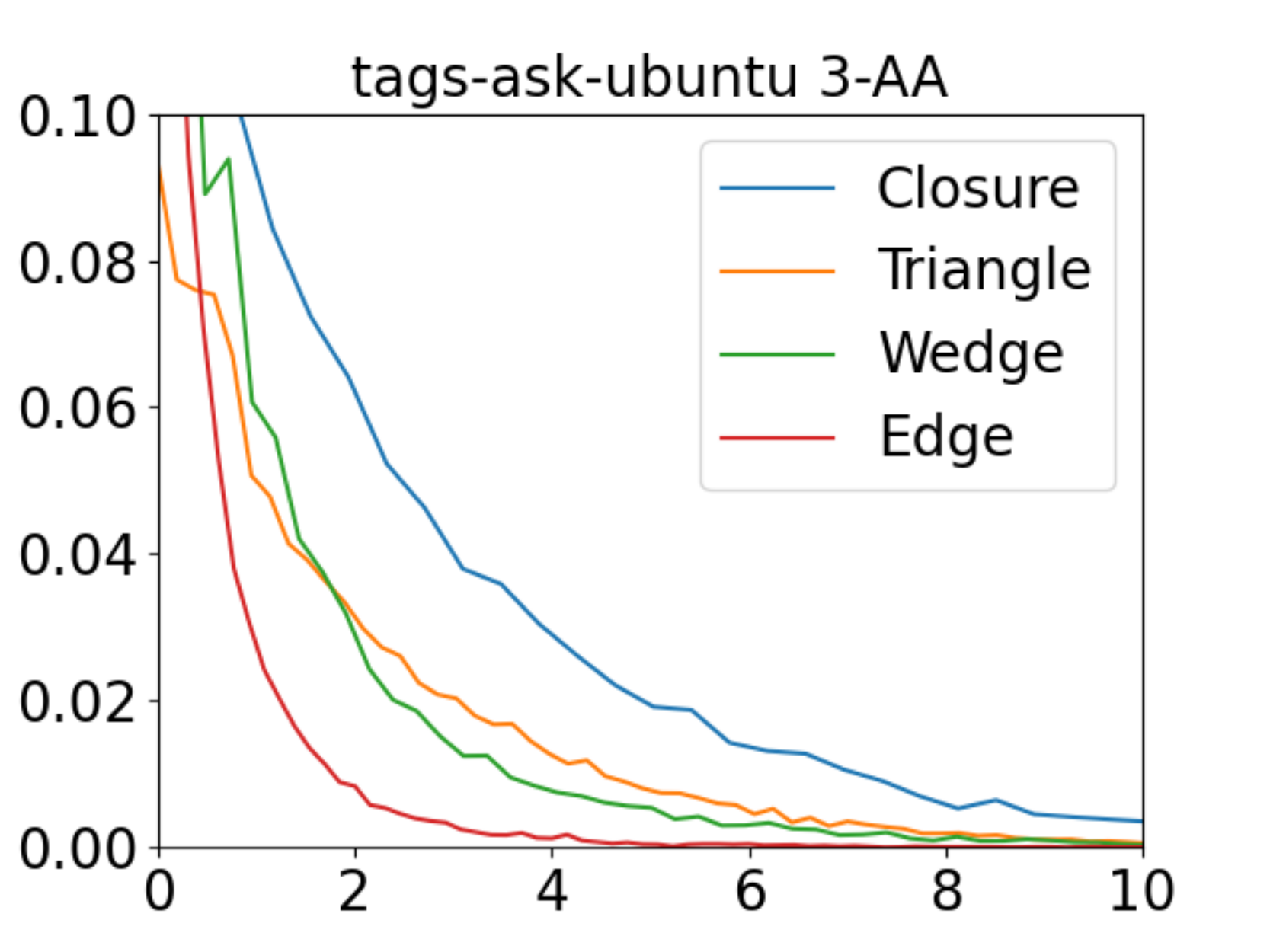}

\includegraphics[trim={0.3cm 0.5cm 0.5cm 0.7cm},clip,width=0.49\textwidth
]{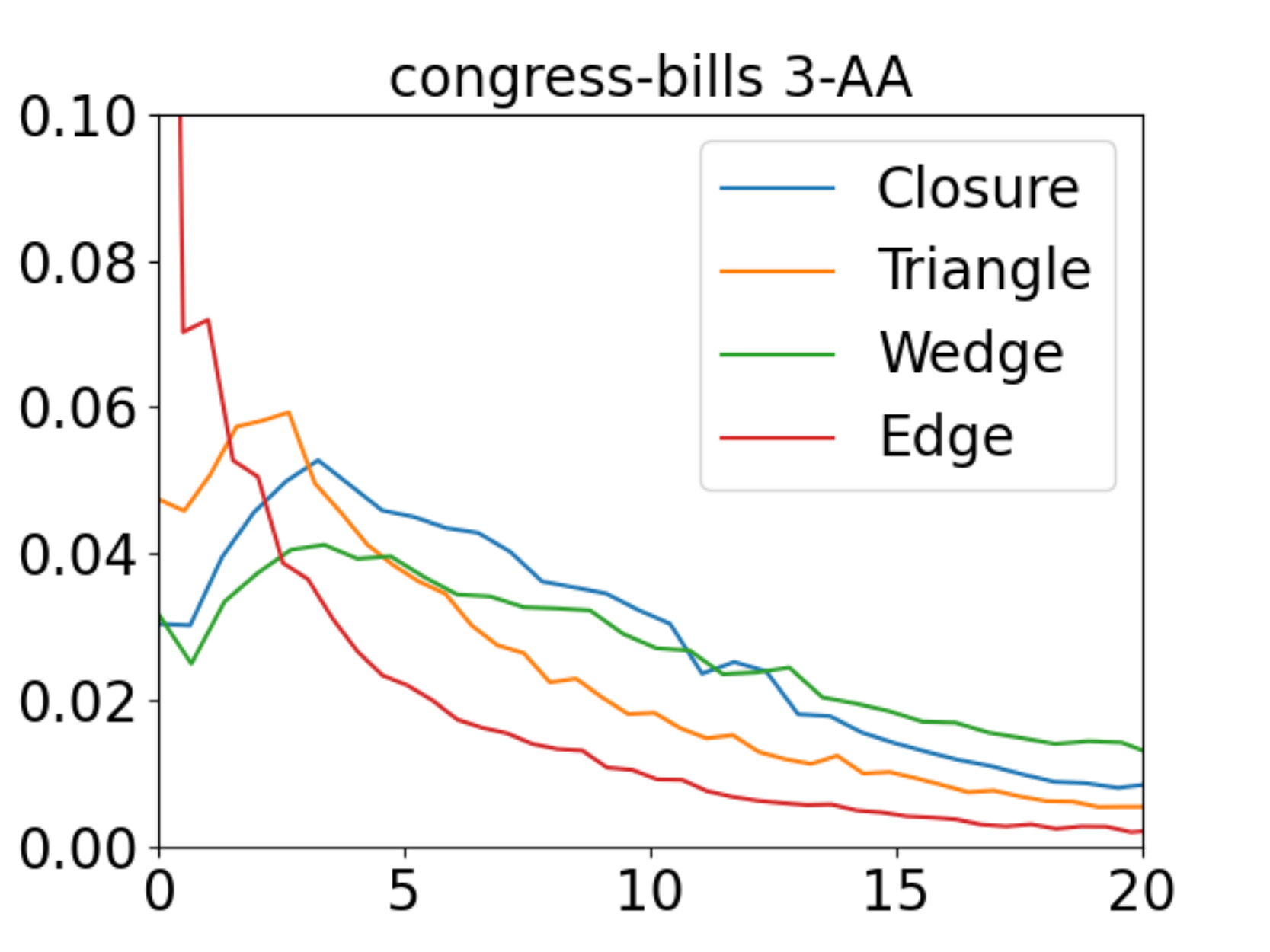}
\includegraphics[trim={0.3cm 0.5cm 0.5cm 0.7cm},clip,width=0.49\textwidth
]{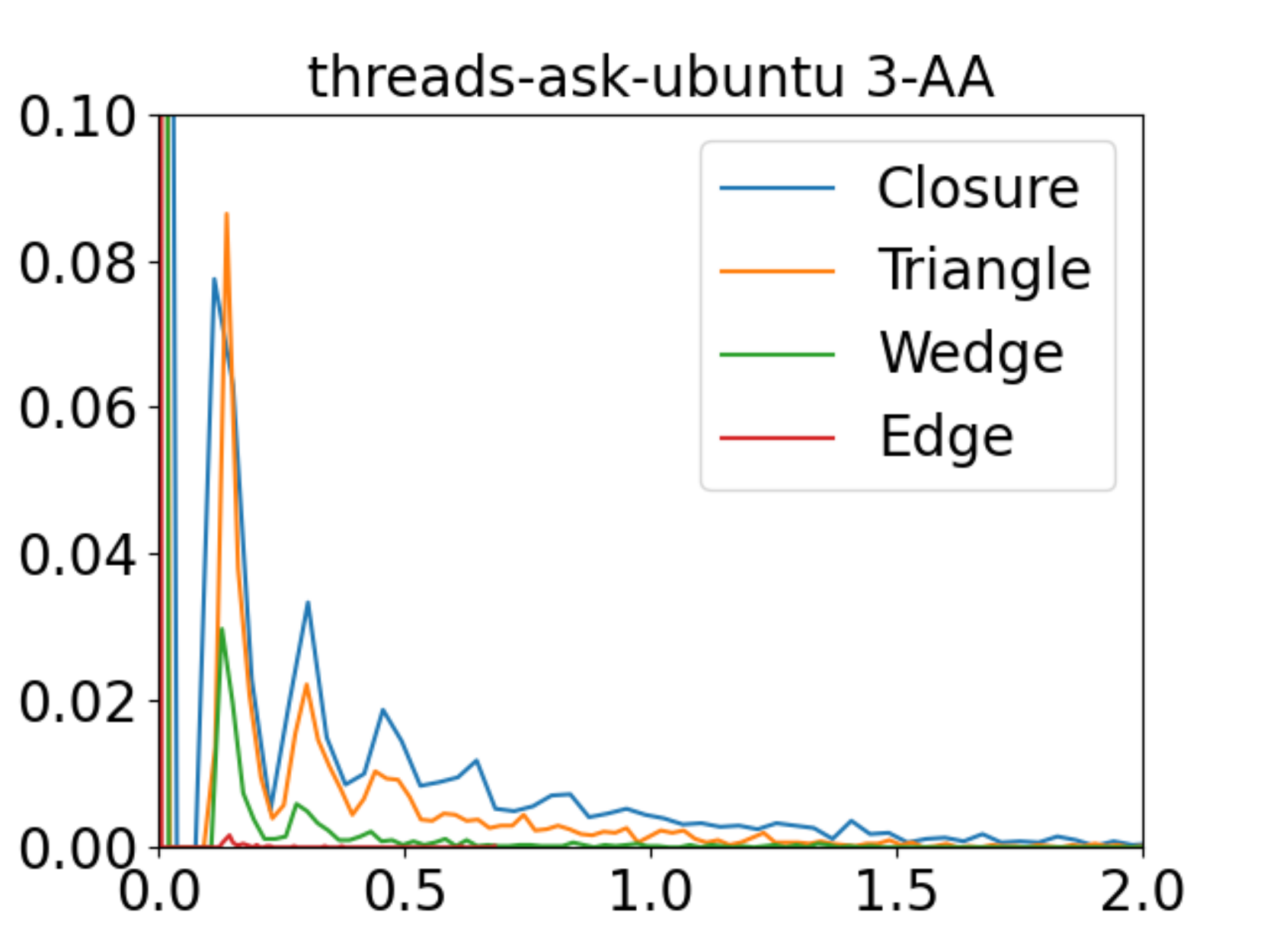}

\vspace{-3mm}
\captionof{figure}{\small{The distributions (PDF) of the arithmetic means of AA's between $(u,v)$, $(u,w)$ and $(v,w)$ for triplets $(\{u,v\},w,t)$ over two datasets DAWN and tags-ask-ubuntu. Distribusions of different classes get severely mixed.}} \label{fig:aa_distribution}
\end{minipage}
\end{table*}

\begin{table*}
     \vspace{-6mm}
    \includegraphics[trim={3cm 1.6cm 4.2cm 1.0cm},clip, height=0.165\textwidth]{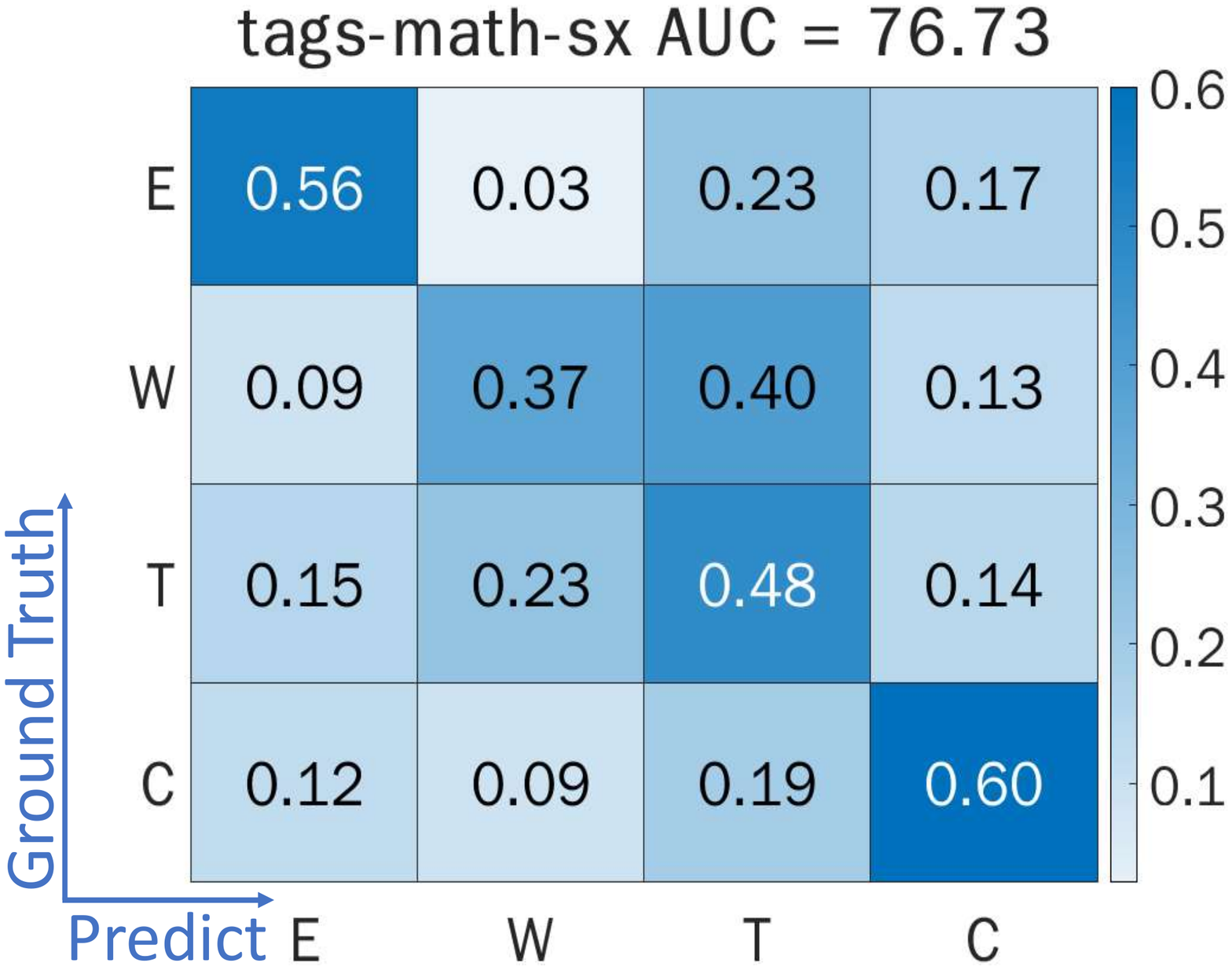}
    \includegraphics[trim={6cm 1.6cm 4.2cm 1.0cm},clip, height=0.165\textwidth]{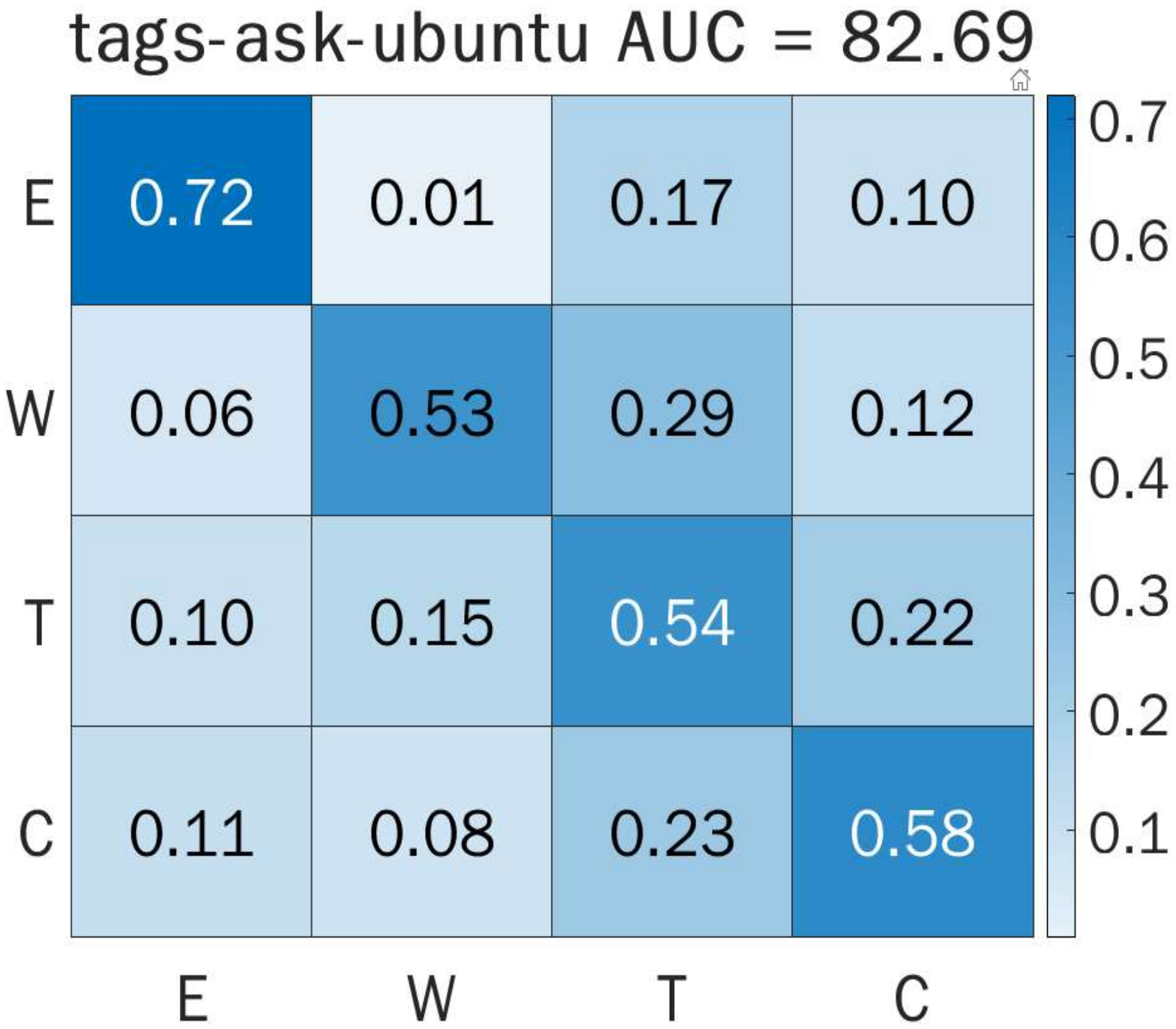}
    \includegraphics[trim={6cm 1.6cm 4.2cm 1.0cm},clip, height=0.165\textwidth]{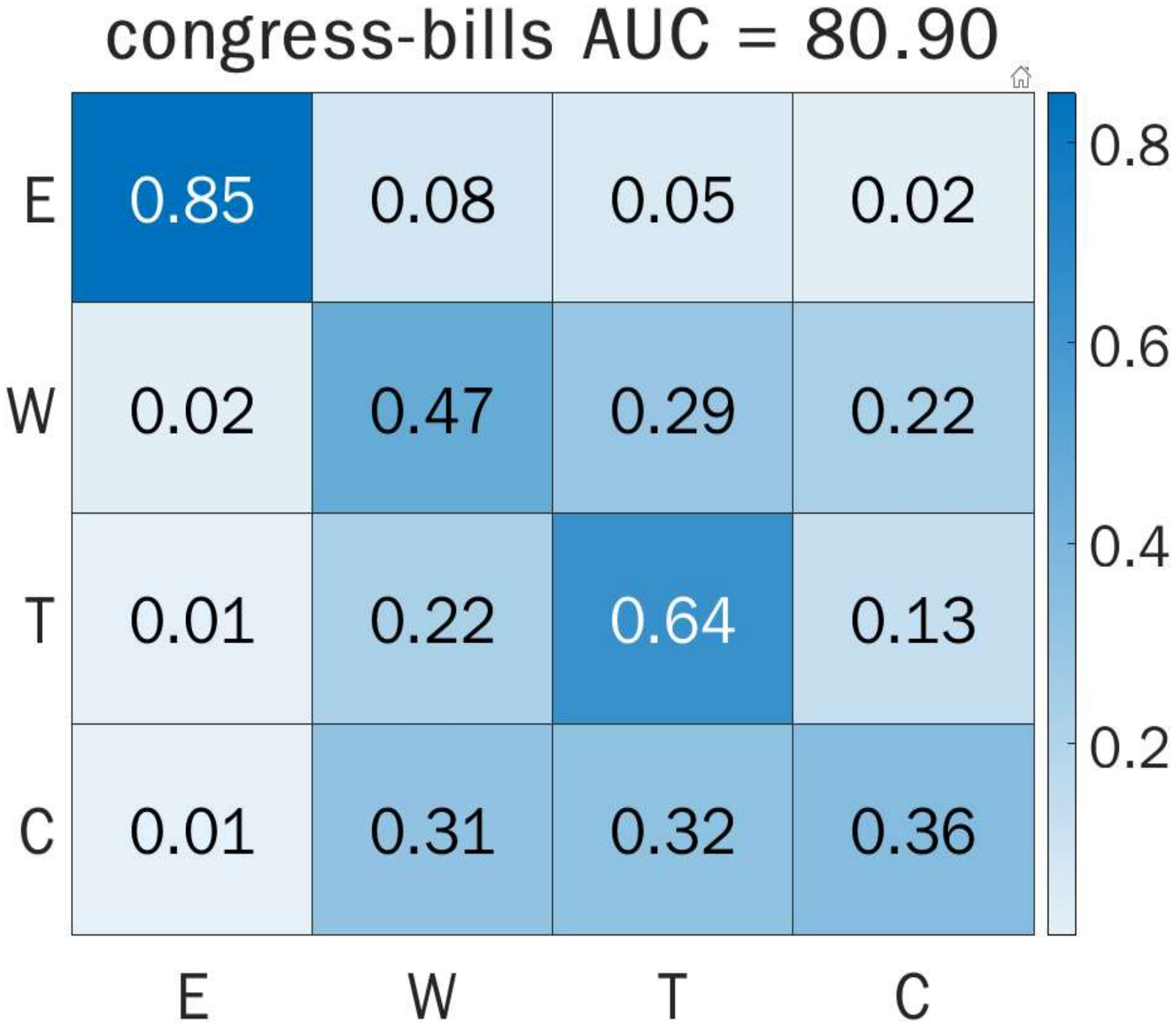}
    \includegraphics[trim={6cm 1.6cm 4.2cm 1.0cm},clip, height=0.165\textwidth]{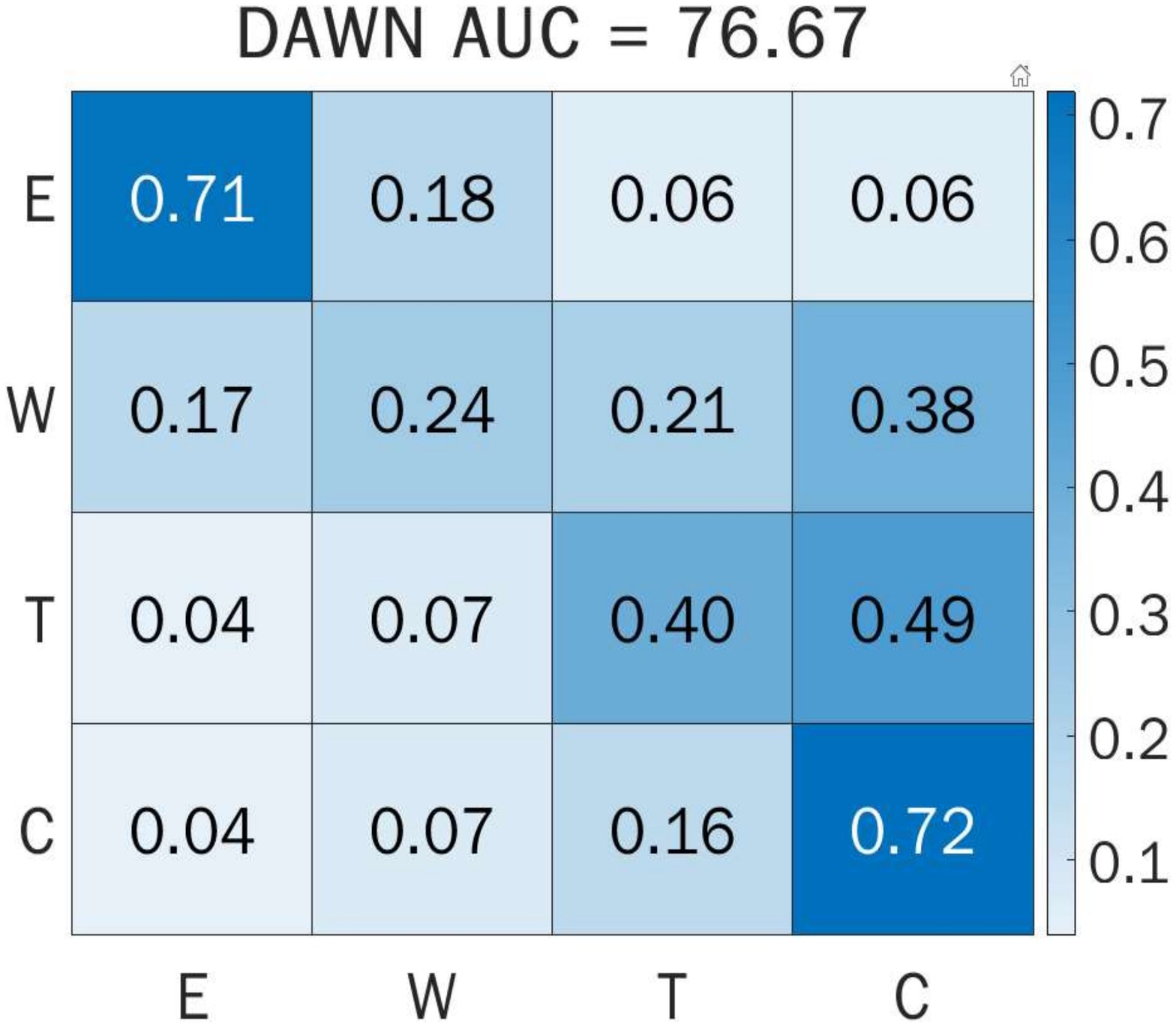}
    \includegraphics[trim={6cm 1.6cm 4.2cm 1.0cm},clip, height=0.165\textwidth]{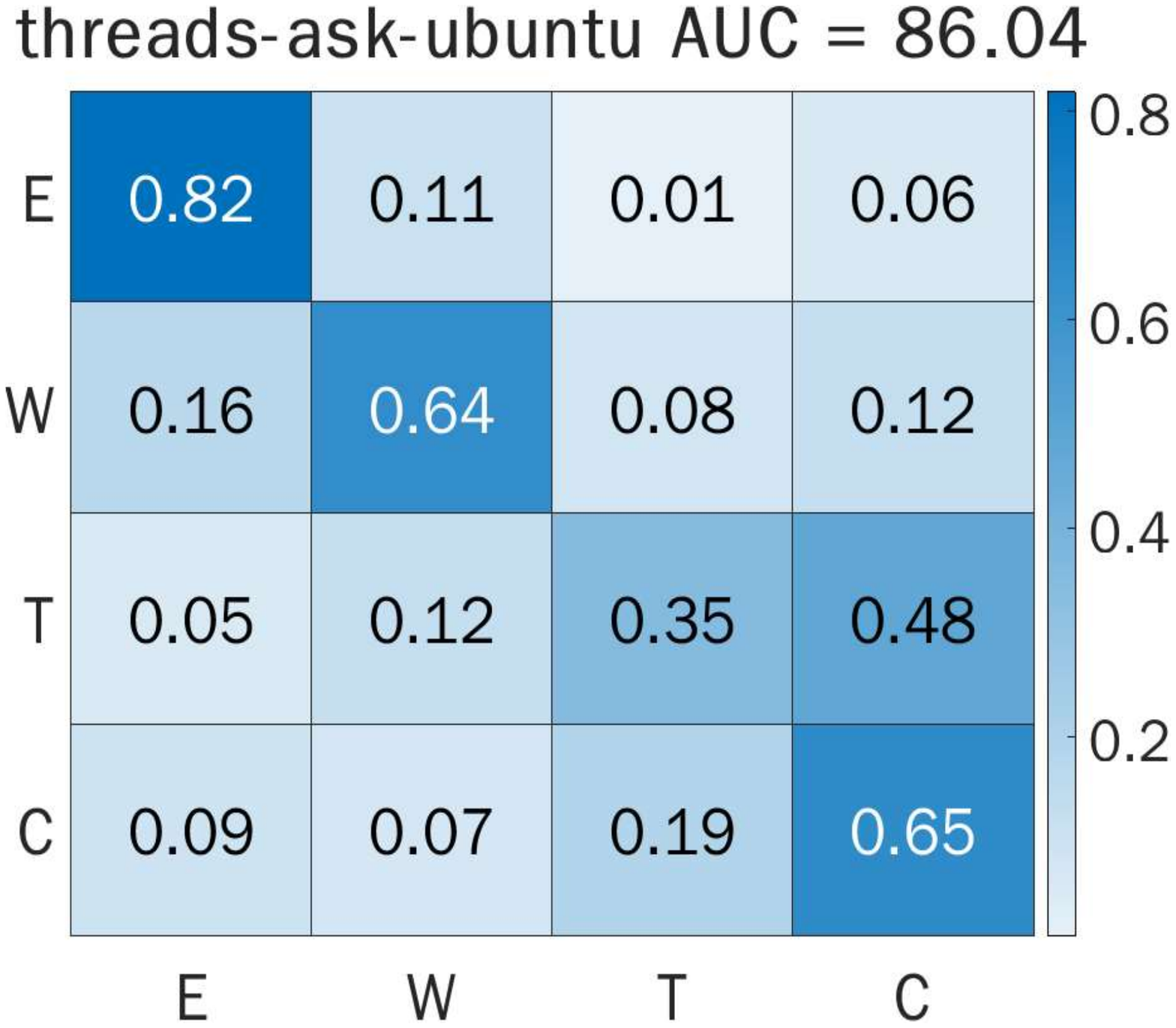}
    \vspace{-1mm}
    \caption{\small{The confusion matrices of \proj to make prediction over $5$ datasets and the corresponding AUC scores. E, W, T and C stand for (E)edge, (W)edge, (T)riangle, and (C)losure. The x-axis is the predicted label, and y-axis is the ground truth. }    }
    \label{tab:confusion_matrix}
    \vspace{-8mm}
    
\end{table*}



\vspace{-2mm}
\subsection{Predicting Higher-order Patterns (for Q1)}
The results are shown in Table~\ref{tab:problem1_performance}. All baselines perform much worse than \proj, which confirms the need for more powerful models to extract subtle structural features to indicate higher-order patterns. 

As the training and test datasets are collected from the networks in different time ranges and do not share the exact same data distribution, some baseline approaches do not generalize well and achieve even less than 0.5 AUCs. 
Previous works~\cite{benson2018simplicial,yoon2020much} demonstrate some possibilities of using heuristic metrics (AA, PA, JC) to distinguish Triangles and Closures in static hypergraphs. However, our results show that these metrics fail to give desired predictions when time constraints are incorporated in the experimental setting. 
We further analyze
the distributions of these metrics and 
find they get severely entangled across different classes in Fig.~\ref{fig:aa_distribution}.
JODIE~\cite{kumar2019predicting} and TGN~\cite{tgn_icml_grl2020} do not work well because they track node embeddings over time which accumulate noise. Such noise does harm to transfer the knowledge of indicative patterns across different node triplets. TGAT~\cite{xu2020inductive}  claims to be inductive and does not accumulate the noise, while it still performs poorly.
There are two main reasons: First, there is a node ambiguity issue in TGAT as illustrated in~\cite{li2020distance,wang2020inductive,zhang2021labeling}. To solve this issue, TGAT needs to have high-quality node or edge attributes, which may be hard to get in practice. Second, TGAT is supervised by lower-order link prediction, which is different from high-order pattern prediction. NHP~\cite{yadati2020nhp} uses network snapshots that lose much temporal information.  

To better interpret the results, we show the confusion matrices of \proj in Table~\ref{tab:confusion_matrix}.
\proj identifies four patterns reasonably well in all these datasets. Interestingly, \proj tends to predict true Edges incorrectly more as Triangles/Closures than as Wedges in the tags datasets. In these datasets, triplets that share similar contextual structures are likely to either evolve to Triangles/Closures in a short time or avoid interaction expansion instead of keeping being Wedges. Another finding is that it is typically hard to distinguish between wedges and triangles. This is because to form a triangle, the triplet must first form a wedge. Thus, more accurate temporal information is needed, \ie, small time granularity with a long enough time window is expected. For some datasets, such as DAWN, have a large time granularity (quarters) and a small time window (4 quarters). It is naturally hard to determine whether the triplet will form a wedge or triangle. Consequently, our model is unable to capture enough temporal information for a better prediction. 

We also show some ablation study of \proj in Table~\ref{tab:ablation}. We demonstrate the significance of time-based TRW sampling. Also, DEs that capture both the symmetry and asymmetry of the patterns perform much better than the sum pooling operation adopted in \cite{li2020distance}, which is crucial to distinguish the patterns such as Wedge \textit{v.s.} Triangle.
\vspace{-2mm}
\begin{table}[h]
\begin{tabular}{c | c | c | c  } 
\toprule  
\small{$\triangle$ AUC} & $\alpha = 0$ & \small{Rv. DE}  & \small{Sym. DE~\cite{li2020distance}} \\
\midrule
tags-ask-ubuntu & -1.59 & -4.42 & -3.13   \\
threads-ask-ubuntu & -2.16 & -3.44 & -2.68  \\ 
\bottomrule 
\end{tabular}\caption{\small{Ablation study. Revise each part of \proj and compare with the full model in Table~\ref{tab:problem1_performance}. ``$\alpha=0$'': Uniformly sample each step of TRWs. ``Rv. DE'': Remove DEs (Eq.~\ref{eq:single-DE}). ``Sym. DE'': Ignore asymmetry between DEs in Eq.~\eqref{eq:DE}. Aggregate three DEs via sum pooling~\cite{li2020distance}.}}
\label{tab:ablation}
\end{table}
\vspace{-8mm}

\begin{table}
\resizebox{\columnwidth}{!}{
\begin{tabular}{ c | c | c c c | c c c  } 

\toprule  
\multicolumn{2}{c|}{ ~ } & 
\multicolumn{3}{c|}{tags-math-sx} & \multicolumn{3}{c}{tags-ask-ubuntu}  \\
\cline{3-8}
\multicolumn{2}{c|}{ ~ }  & \small{Wedge} & \small{Triangle} & \small{Closure} & \small{Wedge} & \small{Triangle} & \small{Closure} \\
\midrule
\multirow{4}*{\small{NLL}}&

\small{TGAT} & 1.54 & 0.77 & 1.37 & 1.55 & 1.03 & 1.40  \\

~ & \small{JODIE} & 1.47 & 2.60 & 2.71 & 1.91 & 0.94 & 1.36 \\


~ & \small{TGN} & 1.23 & 0.69 & 1.45 & 1.50 & 0.87 & 1.37  \\

~ & \small{\textbf{\proj}} & \textbf{1.33} & \textbf{0.68} & \textbf{1.33} & \textbf{1.49} & \textbf{0.88} & \textbf{1.36} 
\\
\midrule
~&

\small{TGAT} & 0.97 & 0.43 & 0.76 & 0.99 & 0.54 & 0.89\\

\small{MAE-} & \small{JODIE} & 1.58 & 1.13 & 1.34 & 1.04 & 0.54 & 0.92 \\


\small{LOG} & \small{TGN} & 0.97 & \textbf{0.40} & 0.73 & 1.00 & 0.50 & 0.88\\

~ & \small{\textbf{\proj}} & \textbf{0.71}  & 0.41 & \textbf{0.95} & \textbf{0.98} & \textbf{0.48} & \textbf{0.83}

\\





\bottomrule 

\end{tabular}
}
\caption{\small{Averaged NLL (Eq.~\eqref{eq:NLL})  and MAE-LOG: $\log t-\hat{\log t}$ where $\log t$, $\hat{\log t}$ are the ground-truth and the estimated log time of time prediction (for Q2): Lower is better. 
}
}\label{tab:time_prediction_performance}
\vspace{-8mm}
\end{table}

\vspace{-1mm}
\subsection{Time Prediction (for Q2)}
\vspace{-1mm}

\proj can predict when a triplet of interest expands the interaction into a certain higher-order pattern for the first time. In this study, we only compare our method with NN-based baselines except NHP because heuristic  and NHP are found not expressive enough to capture elaborate temporal information.
All the baselines share the same 
decoder for Q2 as \proj defined in Sec.~\ref{sec:decoder}. The results in Table~\ref{tab:time_prediction_performance}. We use NLL (Eq.\eqref{eq:NLL}) as the evaluating metrics, which is superior in evaluating how different models learn the underlying distributions of the time. As Table~\ref{tab:time_prediction_performance} shows, \proj estimates the distributions of the time uniformly more accurately than baselines. We also compare
MAE-LOG, which provides an explicit characterization of the gap between this estimated log time and the ground truth. 
By analyzing MAE-LOG, \proj is expected to yield a time estimation $\hat{t}\in[2.29^{-1}t, 2.29t]$ where $t$ is the ground truth. In contrast, baselines yields a estimation in $[2.41^{-1}t, 2.41t]$ and \proj reduces the approximation factor of the estimation by about 5\%.

\proj does not outperform the baselines as significantly in Q2 as in Q1. The main reason is that Q2 is much easier for baselines than Q1. For Q1, baselines need to distinguish those hierarchical high-order patterns. As they either accumulate noise or have node ambiguity issues, they generally find it hard to predict the patterns accurately. For Q2, instead, we simplify the task by restricting the time prediction within one single type of high-order patterns.

\vspace{6mm}

\begin{figure}
 \begin{minipage}{0.31\textwidth}
 \centering
\resizebox{\textwidth}{!}{
\begin{tabular}{c c}

\toprule Tasks (AUC) & Closure \textit{v.s.} Triangle (65.63) \\
\midrule
\multirow{2}*{Largest $C_W$} & (\{$x$, $x$\}, 0) (\{1, $x$\}, 1) (\{$x$, $x$\}, 2) $100.00\%$ \\
~ & \textcolor{blue}{(\{0, $x$\}, $x$) (\{1, $x$\}, 1) (\{2, $x$\}, $x$)} $100.00\%$  
\\
\midrule
\multirow{2}*{Smallest $C_W$} & (\{0, $x$\}, $x$) (\{1, $x$\}, $x$) (\{1, $x$\}, $x$) $38.94\%$ 
\\
~ & (\{0, $x$\}, $x$) (\{1, $x$\}, $x$) (\{2, 2\}, $x$) $54.49\%$  
\\
\midrule

Tasks (AUC) & Closure and Triangle \textit{v.s.} Wedge (57.58) \\

\midrule

\multirow{2}*{Largest $C_W$} & (\{0, $x$\}, $x$) (\{1, $x$\}, $x$) (\{2, $x$\}, 1) $84.26\%$ \\

~ & (\{0, $x$\}, $x$) (\{1, $x$\}, $x$) (\{2, $x$\}, 2) $66.13\%$ \\

\midrule

\multirow{2}*{Smallest $C_W$} & \textcolor{blue}{(\{$x$, 0\}, $x$) (\{$x$, 1\}, $x$) (\{$x$, 1\}, $x$)} $34.68\%$ \\

~ & \textcolor[RGB]{34, 139, 34}{(\{0, $x$\}, $x$) (\{1, $x$\}, $x$) (\{2, $x$\}, $x$)} $49.74\%$ \\

\midrule

Tasks (AUC) &  Wedge \textit{v.s.} Edge (71.10)\\

\midrule

\multirow{2}*{Largest $C_W$} & (\{$x$, 0\}, $x$) (\{$x$, 1\}, $x$) (\{$x$, 2\}, 2) $49.14\%$  \\

~ & (\{$x$, 0\}, $x$) (\{$x$, 1\}, $x$) (\{$x$, 2\}, $x$) $50.26\%$ \\

\midrule

\multirow{2}*{Smallest $C_W$} & \textcolor{blue}{(\{0, $x$\}, $x$) (\{$x$, 1\}, $x$) (\{1, 2\}, $x$)} $24.60\%$ \\

~ & \textcolor[RGB]{34, 139, 34}{(\{$x$, 0\}, $x$) (\{$x$, 1\}, $x$) (\{2, 2\}, $x$)} $28.81\%$
\\
\bottomrule 

\end{tabular}
}

\end{minipage}
\begin{minipage}{0.16\textwidth}
 
 \vspace{0.1mm}
 
 \includegraphics[trim={10.6cm 8.1cm 10.6cm 7.6cm},clip, width=1.\textwidth]{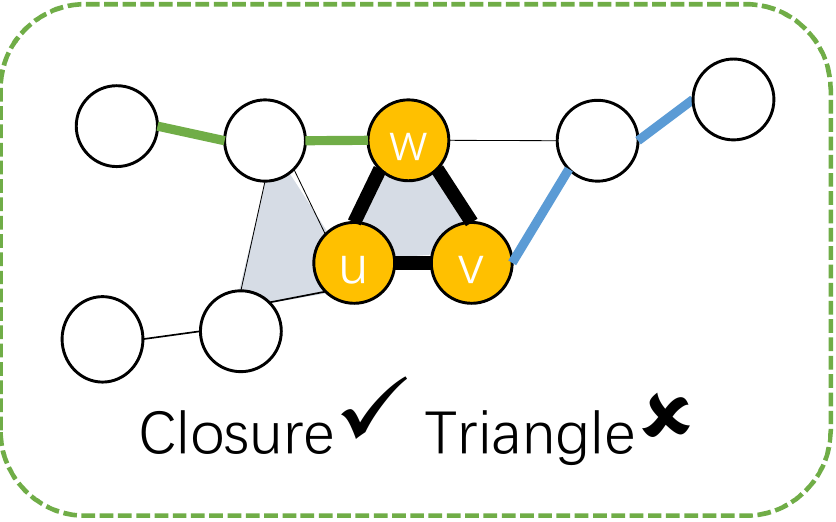}
 
 \vspace{0.5mm}
 
 \includegraphics[trim={10.3cm 8.8cm 10.9cm 6.9cm},clip, width=1\textwidth]{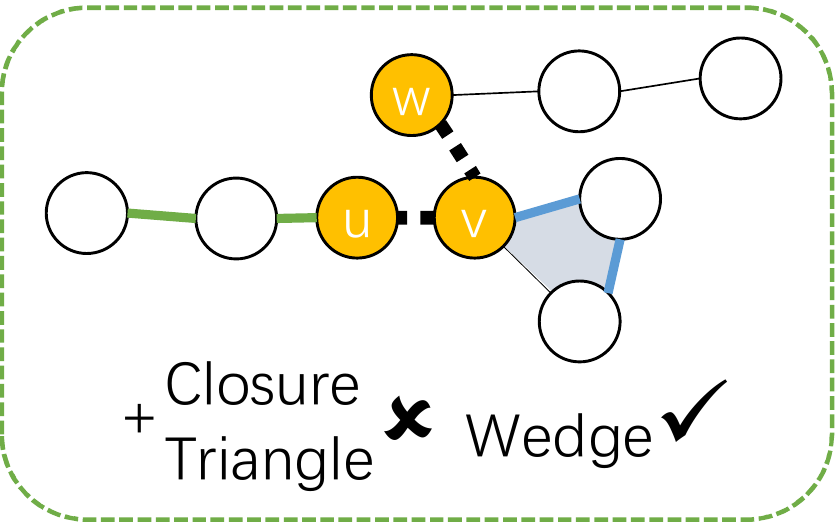}
 
 \vspace{0.5mm}
 
 \includegraphics[trim={10.55cm 8.1cm 10.65cm 7.6cm},clip, width=1.0\textwidth]{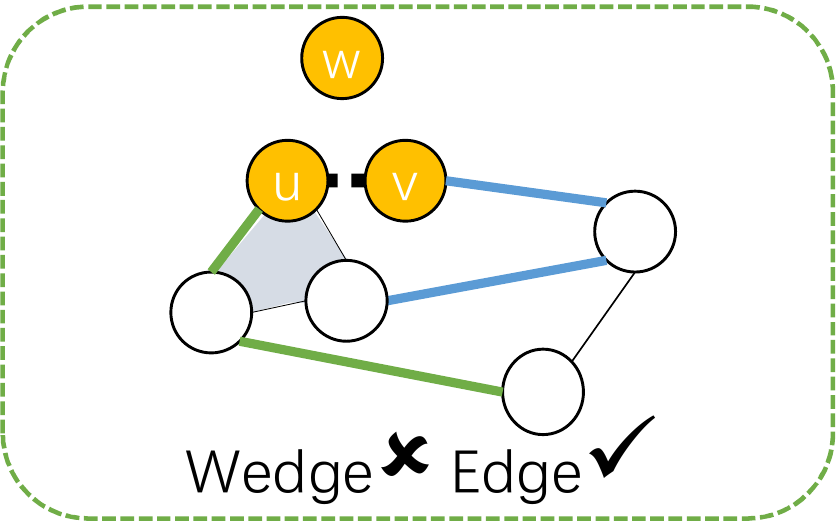}
\end{minipage}
\captionof{figure}{\small{The above table is the TRWs with the top-2 largest and smallest scores $C_W$ (Eq.~\eqref{eq:TRW-score}) to distinguish two patterns in tags-ask-ubuntu. TRWs are categorized based on DEs: Given a node triplet $(\{u,v\},w,t)$, a TRW $W$ is categorized based on $W = (I(W[0]),I(W[1]),I(W[2]))$ where $I(a)=(\{\tilde{g}(a;S_u),\tilde{g}(a;S_v)\},\tilde{g}(a;S_w))$, $\tilde{g}(a;S_z)$ uses SPD 
and $x$ indicates SPD$\geq3$. \proj uses the decoder for Q3. The ratio besides each TRW refers the number of the corresponding TRWs that appear as features for the first class over that for both classes. The figures on the right are the illustration of TRWs based on those \textcolor{blue}{blue} or \textcolor[RGB]{34, 139, 34}{green} TRWs in the left table.
}} 
\vspace{-6mm}
\label{tab:interpretation} 
\end{figure}

\vspace{-8mm}
\subsection{Most Discriminatory TRWs (for Q3)}
Searching for the important features that induce certain prediction is always a challenging problem in machine learning, especially for NN-based models. Identifying network structural features is even more challenging as they are irregular and strongly correlated by nature. 
\proj provides possibilities to identify the most discriminatory structural features in the format of TRWs. As illustrated in Sec.\ref{sec:decoder} the \emph{decoder for Q3} paragraph, \proj allows categorizing TRWs into different types and summarized the impact of each individual type as a single score on the higher-order pattern prediction. 

We evaluate \proj over tags-ask-ubuntu.
We rank different types of TRWs based on their discriminatory scores $C_W$ (Eq.\eqref{eq:TRW-score}) and list the top-2 ranked TRWs in Table~\ref{tab:interpretation}. Because the NN-based encoder encodes each TRW separately to achieve $C_W$ and the decoder is simply a linear logistic regression, \proj yields very stable model interpretations. The top-2 ranked TRWs always stay in the top-5 ranked ones in the five random experiments.

We sample $3,000$ examples for each task and compute the ratio of each type of TRWs that appear as the features for the first class and both two classes (first + second). We see a good correlation between the discriminatory score $C_W$'s and the ratios, which indicates the effectiveness of \proj in selecting indicative TRW features.

Those most discriminatory TRW features shown in Table~\ref{tab:interpretation} demonstrates interesting patterns. We further study these patterns by using three examples as illustrated in 
Fig.~\ref{tab:interpretation}
(TRWs with blue or green font in 
the table 
are drawn). 



\begin{itemize}[leftmargin=*]
\item \emph{Closures v.s. Triangles}: They start from $u,v$ (or $w$) and jump to some nodes that are directly connected to $w$ (or $u,v$). This means that $u,v,w$ tend to form a Closure pattern if both $u,w$ and $v,w$ have common neighbors before. We may interpret it as a higher-order pattern closure that generalizes triadic closure in graphs~\cite{simmel1950sociology}. Without hypergraph modeling, we cannot distinguish Closures and Triangles, and thus cannot observe this higher-order pattern closure.
\item \emph{Wedges v.s. Closures + Triangles}: They start from $u$ (or $v$, $w$) and jump to some nodes that are still far away from the other two nodes. These TRWs can be viewed as the opposite of the above higher-order pattern closure. An interesting question is whether they are the most indicative ones for Edges. See the next bullet. 
\item \emph{Edges v.s. Wedges}: Interestingly, they follow a uniform pattern by jumping around $u$ and $v$. Therefore, the most indicative features that forbid any interaction expansion are the structures that densely connect $u$ and $v$ but disconnect $w$. Meanwhile, combined with the previous bullet, the most indicative features for interaction expansion but from only one node (Wedges) are the structures that connect a single node in the node triplet.

\end{itemize}

\subsection{Parameter sensitivity}
\label{sup:para_sensitivity}

We study the sensitivity of \proj w.r.t. three hyperparameters the number of total walks $M$, the length of each walk $m$, and $\alpha$ in Alg.~\ref{alg:TRW} which indicates how likely \proj is to sample the hyperedges that appeared a long time ago. 
Results are shown in Fig.~\ref{fig:abla_study}.
Under most circumstances, the performance of $M=64$ is higher than $M=32$ and $M=16$. This indicates that a larger $M$ can ensure good performance. 
Fig.~\ref{fig:abla_study} b) shows that with the fixed $M=64$, different $m$'s yield similar results as long as the first-hop neighbors ($m\geq 1$) get sampled. For more complicated networks, we expect a larger $m$ may be needed.  
The sensitivity study of $\alpha$ is shown in Fig.~\ref{fig:abla_study} c). We choose $\alpha$ from $1e-8$ to $1e-5$. The choice of $\alpha$ should be normalized w.r.t. the average edge intensity that indicates how many edges appear within a time unit.  The average edge intensity of a network is $2|\mathcal{E}||e|^2/(|V|T)$, where $|\mathcal{E}|$ is the number of hyperedges, $|e|$ is the average size of hyperedge, $|V|$ is the number of nodes, $T$ is the entire time range. We normalize the average edge intensity to $1e-5$ per time unit by rescaling the entire time range $T$. And thus, the ratio between the average edge intensity and $\alpha$ is between $1$ and $1000$.
For threads-ask-ubuntu, the case $\alpha=1e-6$ slightly outperforms the other three cases while in general different $\alpha$'s give comparable results. For tags-ask-ubuntu, a too small $\alpha$ may introduce performance decay, which demonstrates that more recent hyperedges may provide more informative patterns.

\vspace{-3mm}
\begin{figure}[h]

\includegraphics[trim={3.4cm 9.2cm 4.0cm 9.5cm},clip,width=0.152\textwidth]{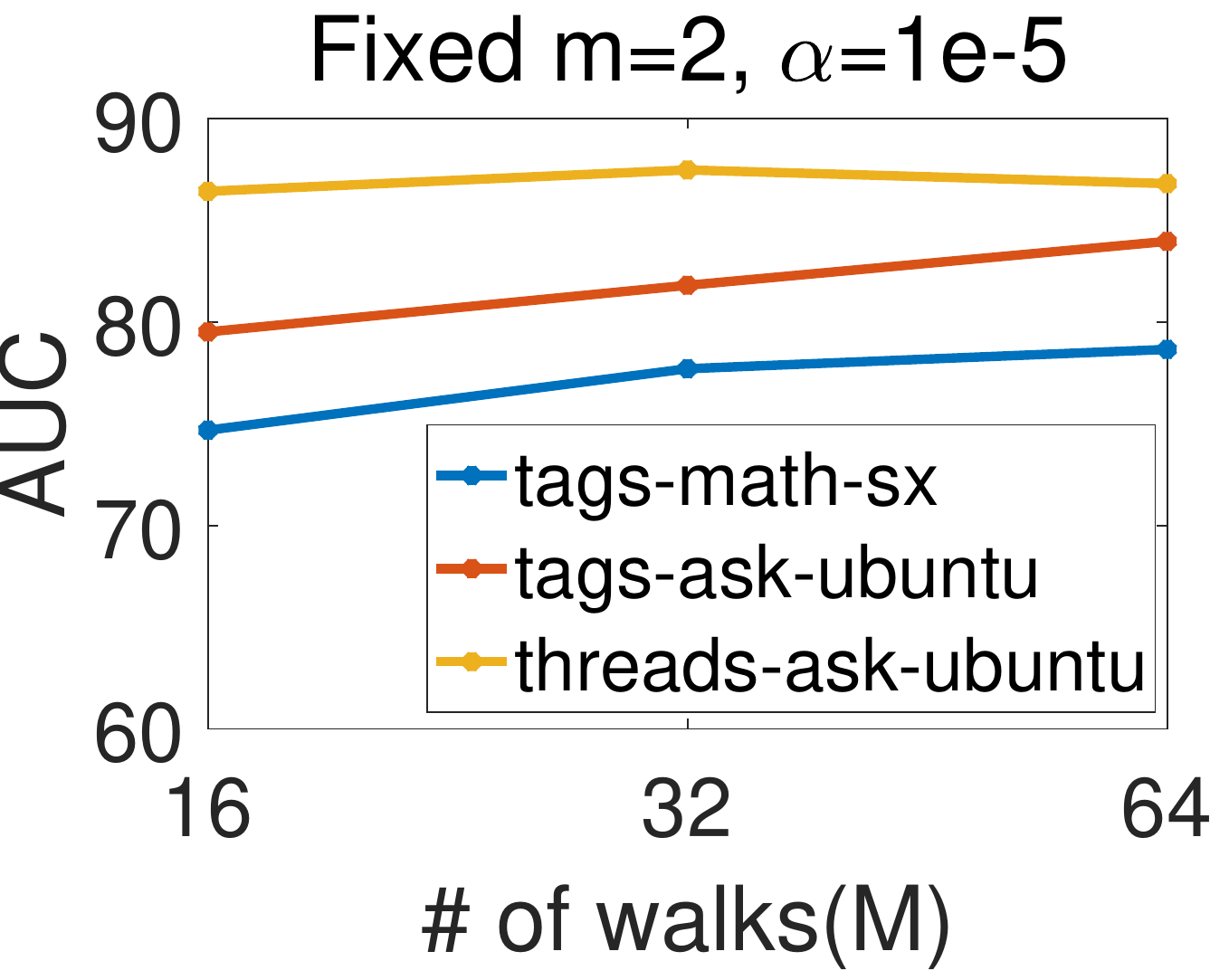}
\includegraphics[trim={3.4cm 9.2cm 4.0cm 9.5cm},clip,width=0.152\textwidth]{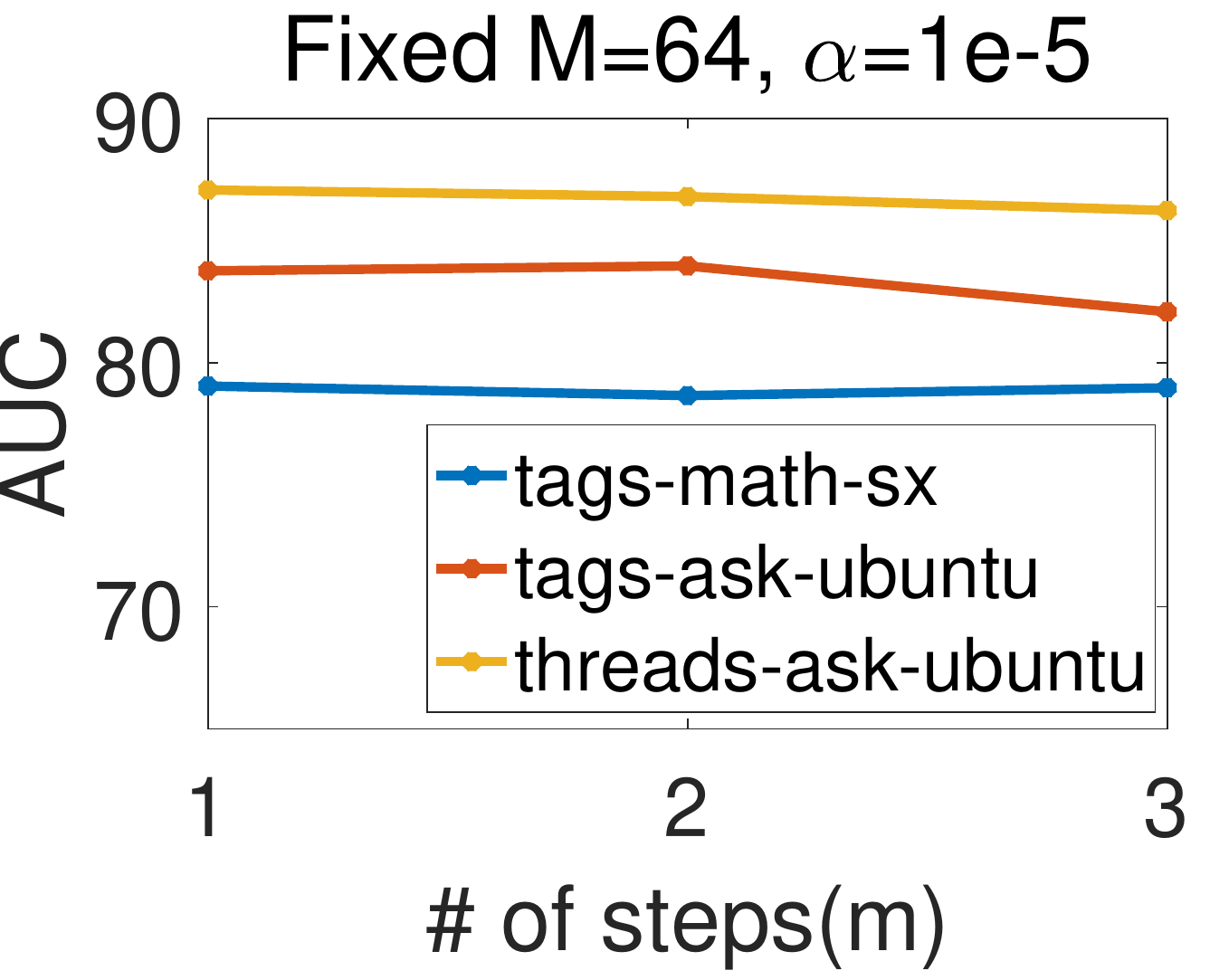}
\includegraphics[trim={3.4cm 9.2cm 3.6cm 9.5cm},clip,width=0.16\textwidth]{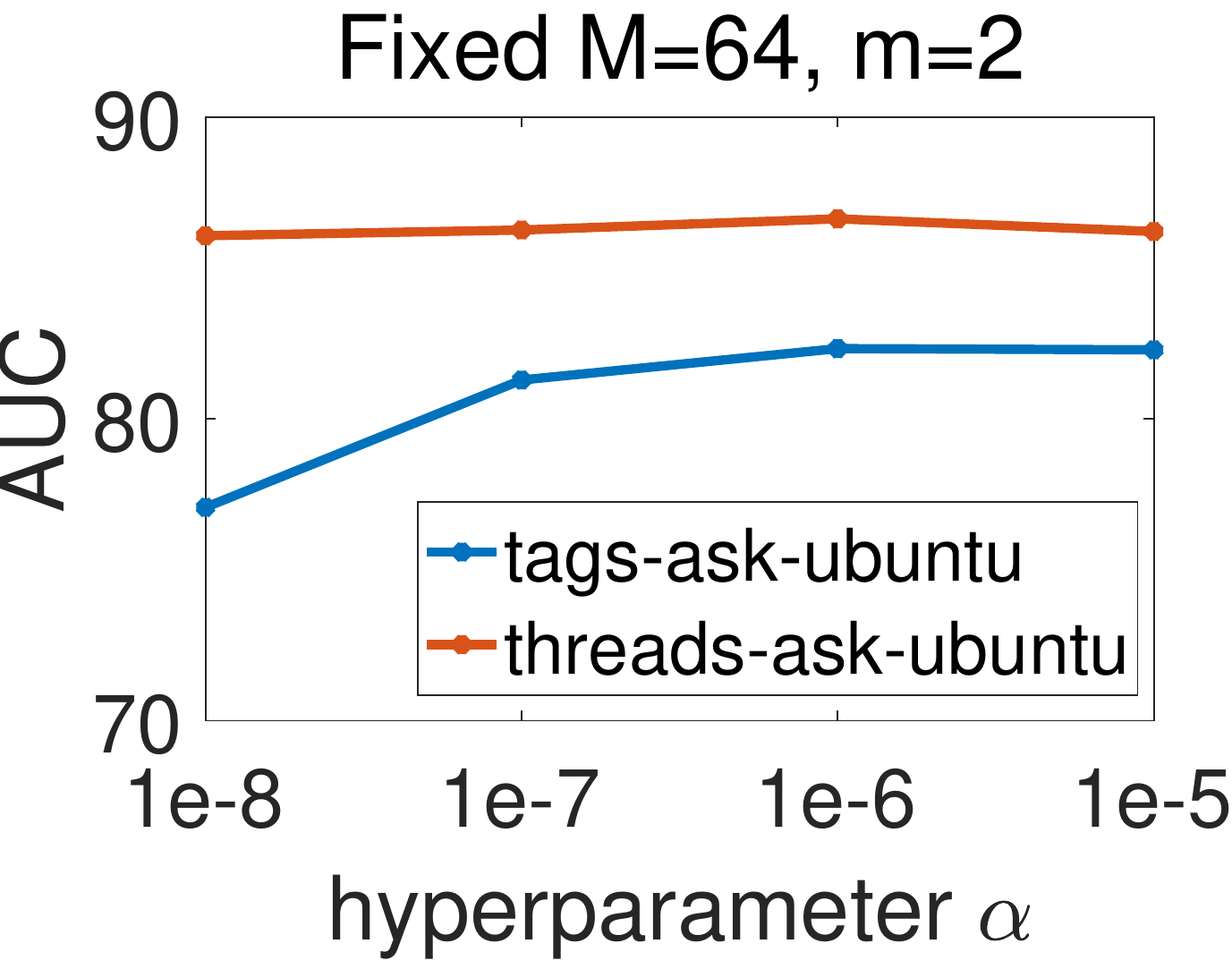}

\vspace{-3mm}
\captionof{figure}{\small{
The ablation study of $M$, $m$, and $\alpha$ for different datasets. We set every other parameter the same and only change the $M$, $m$ or $\alpha$, and report the AUC for Q1. 
}} \label{fig:abla_study}
\end{figure}
\vspace{-5mm}

\section{Conclusion and Future Works}
\label{sec:conclusion}

In this work, we proposed the first model \proj to predict higher-order patterns in temporal hypergraphs to answer what type of, when, and why interactions may expand in a node triplet. \proj can be further generalized to predict even higher-order patterns. In the future, for applications, \proj may be adapted to study higher-order pattern prediction in scientific disciplines such as biology~\cite{zhang2019hyper,alsentzer2020subgraph}. Methodologically, \proj may be further generalized to predict the entire lifecycles of higher-order patterns as shown in Fig. 5 of \cite{benson2018simplicial}.

\bibliographystyle{ACM-Reference-Format}
\bibliography{reference}

\section*{Acknowledgments}
We greatly thank the actionable suggestions given by reviewers. Y. Liu and P.L. are supported by the 2021 JPMorgan Faculty Award and the National Science Foundation (NSF) award HDR-2117997.

\appendix
\begin{center}
\Large \textbf{Appendix}
\end{center}

\vspace{-2mm}
\section{Generalization of Asymmetric DE}
\label{sec:DE-explain}

As discussed in the main text,
if two node triplets share the same historical contextual network structure, then any nodes that appear on TRWs sampled for these two node triplets hold the same DE distribution. 
Formal mathematical characterization is given as follows. First, we define what is shortest path distance between hyperedges.

\begin{definition} [Shortest path distance between hyperedges]
The shortest path distance between two hyperedges $e$ and $e'$ is the minimal number of hyperedges $e_1, e_2,...,e_k$ where $e_1=e$, $e_k=e'$ and $e_i\cap e_{i+1}\neq \emptyset$ for $i=1,2,...,k-1$.
\end{definition}

Based on this definition, we give the following definition of historical $m$-hop subgraphs.

\begin{definition} [Historical $m$-hop subgraphs]
Given a node triplet $\tau = (\{u, v\}, w, t)$, consider the set of hyperedges that appear before $t$ and hold a shortest path distance at most $m$ to either $u$, $v$ or $w$. Denote this set of hyperedges as $\mathcal{E}_{\tau,m}$. The historical $m$-hop subgraph of $\tau$, $G_{\tau,m}$ is the subgraph that consists of the nodes covered by the hyperedges in $\mathcal{E}_{\tau,m}$ and the hyperedges in $\mathcal{E}_{\tau,m}$. 
\end{definition}


Next, we define subgraph isomorphism between two historical $m$-hop subgraphs. 
\begin{definition} [Subgraph isomorphism]
Consider two triplets $\tau_1 = (\{u, v\}, w, t)$ and $\tau_2 = (\{u', v'\}, w', t')$. Denote their historical $m$-hop subgraphs as $G_{\tau_1,m}$ and $G_{\tau_2,m}$ respectively. $G_{\tau_1,m}$ and $G_{\tau_2,m}$ are called isomorphic if there exists a bijective mapping $\pi$ from the node set of $G_{\tau_1,m}$ to the node set of $G_{\tau_2,m}$ such that for every temporal hyperedge $(\tilde{e},\tilde{t})$ in $G_{\tau_1,m}$, there is a corresponding temporal hyperedge $(\tilde{e}',\tilde{t}')$ in $G_{\tau_2,m}$ where the time differences are same $t-\tilde{t} = t' - \tilde{t}'$, and hyperedges of two subgraphs hold the correspondence $\tilde{e}' =\{\pi(v)|v\in \tilde{e}\}$. Moreover, the two node triplets hold the correspondence that $\{\pi(u), \pi(v)\} = \{u',v'\}$, $\pi(w) = w'$.
\end{definition}

Finally, generalization/transferability of DE is defined. 
\begin{theorem} [DE is transferrable/generalizable]\label{thm:gen}
Consider two triplets $\tau_1 = (\{u, v\}, w, t)$ and $\tau_2 = (\{u', v'\}, w', t')$. Suppose the historical $m$-hop subgraphs $G_{\tau_1,m}$ and $G_{\tau_2,m}$  of $\tau_1$ and $\tau_2$ are isomorphic where the bijective mapping between the two node sets that preserves such isomorphism is $\pi$. 
Denote the sampled sets of $m$-step temporal random walks for these two triplets as $S_1 = S_u\cup S_v\cup S_w$ and $S_2 = S_{u'}\cup S_{v'}\cup S_{w'}$. The following two statements are true: 
\begin{enumerate}[leftmargin=*]
    \item  For any node $a$ in $G_{\tau_1,m}$ and thus $\pi(a)$ is in $G_{\tau_2,m}$, the probability that $a$ appears in one TRW of $S_1$ is the same as the probability that $a'=\pi(a)$ appears in one TRW of $S_2$.
    \item The above two nodes $a$ and $a'=\pi(a)$ that appear in TRWs have the same conditional probabilies of their DEs, \ie, \[\mathbb{P}(I(a|\{S_u, S_v, S_w\})| a \,\text{in $S_1$}) = \mathbb{P}(I(a'|\{S_{u'}, S_{v'}, S_{w'}\})| a' \,\text{in $S_2$}).\]
\end{enumerate}
\end{theorem}
\begin{proof}
First, node identities are ignored when calculating the distance encoding $I(\cdot)$ and sampling TRW (Algorithm~\ref{alg:TRW}). Only the relative position of a node in $G_{\tau,m}$ is considered.
Since $G_{\tau_1,m}$ and $G_{\tau_2,m}$ are isomorphic, the probability to sample a TRW 
$W$ from $G_{\tau_1,m}$, where $W=((z_0, t_0), (z_1, t_1), ...,(z_m, t_m))$ is equal to the probability that sampling the corresponding TRW from $G_{\tau_2,m}$ under the mapping $\pi$, \ie,  $\pi(W)=((\pi(z_0), t'-t+t_0), (\pi(z_1), t'-t+t_1), ...,(\pi(z_m),t'-t+ t_m))$. 
Furthermore, because each TRW is sampled independently, for any node $a$ in $G_{\tau_1,m}$, the probability that node $a$ appears in the $S_1$ must equal the probability that $\pi(a)$ appears in the $S_2$. So we have proved the statement 1.

The above analysis also indicates that the probability that we sample three sets of TRWs $(\{S_u, S_v\}, S_w)$ from $G_{\tau_1,m}$  is also equal to the probability to sample three sets of TRWs $(\{S_{u'}, S_{v'}\}, S_{w'})$ from $G_{\tau_2,m}$ as long as $\{S_{u'}, S_{v'}\}=\{\pi(S_{u}), \pi(S_{v})\}$ and $ S_{w'} = \pi(S_{w})$. Here $\pi(S)$ denotes the set $\{\pi(W)|W\in S\}$.

Furthermore, given $\{S_{u'}, S_{v'}\}=\{\pi(S_{u}), \pi(S_{v})\}$ and $ S_{w'} = \pi(S_{w})$, we have $I(a|\{S_u, S_v, S_w\}) = I(a'|\{S_{u'}, S_{v'}, S_{w'}\})$ if $a' = \pi(a)$ because of the definition of DE (Eq.\eqref{eq:DE}). 

Combining the above two result, we have statement ii).
\end{proof}
\vspace{-2mm}
Theorem~\ref{thm:gen} directly indicates that as long as two triplets show the same historical $m$-hop subgraphs, the distributions of the DEs computed based on sampled TRWs must be equal in the sense of probability. This guarantees that DEs will has good generalization capability and transferability. 

\vspace{-2mm}
\section{Experiments setting and baselines}
\label{sup:experiments_baselines}

\subsection{Environment}
\label{sup:experiment_env}
We ran all of our experiments on a server with CPU Intel(R) Xeon(R) Gold 6248R CPU @ 3.00GHz. The GPUs, one is Quadro RTX 6000. The code can be found \href{https://github.com/Graph-COM/Neural_Higher-order_Pattern_Prediction}{\textbf{https://github.com/Graph-COM/Neural\_ Higher-order\_Pattern\_Prediction}}. 

\vspace{-2mm}
\subsection{Dataset description} \label{sec:data}
We use five datasets for evaluation. They are Tags-math-sx, Tags-ask-ubuntu, Congress-bills, DAWN, and Threads-ask-ubuntu. All datasets can be found from \href{https://www.cs.cornell.edu/~arb/data/}{\emph{https://www.cs.cornell.edu/$\sim$arb/data/}}.

\begin{itemize}[leftmargin=*]
    \item \textbf{Tags-math-sx} is a collection of Mathematics Stack Exchange. Users can post questions on the website and attach up to $5$ relative mathematics areas as tags to each question. We denote the areas as nodes and each question as a hyperedge. The timestamps were recorded when questions were posted. 
    \item \textbf{Tags-ask-ubuntu} is similar to Tags-math-sx but it is collected from the website Ask Ubuntu. The nodes are the areas related to computer science and the questions posted by the users are the hyperedges connecting different computer science areas.  
     \item \textbf{Congress-bills} is a dataset of the bills from sponsors and co-sponsors. These bills are put forth in Congress. We denote the sponsors and co-sponsors as the nodes and the bills as hyperedges. The timestamps are recorded when the bills are introduced. 
    \item \textbf{DAWN} records the drug information taken by the patients. Each patient reports the drugs they used before when they visit the emergency department. Here, the drugs are nodes and patients are hyperedges.  The timestamps are recorded when patients visit the emergency  department. 
    \item \textbf{Threads-ask-ubuntu} shows the relationship between the users and threads from the website Ask Ubuntu. 
    The nodes are the users. The users participating in a thread form a hyperedge. 
    The timestamps are recorded when the thread is posted.
    
\end{itemize}


\subsection{Dataset proprocessing} \label{sec:proc}
 There are two different settings in the $5$ datasets. For tags-math-sx, tags-ask-ubuntu, and threads-ask-ubuntu the timestamps are recorded at millisecond resolution. 
For congress-bills and DAWN, the timestamps are recorded at day  and quarter resolution respectively. 
For the timestamps in all datasets, we set the initial times of them to $0$ by subtracting the minimum timestamp in that dataset. For tags-ask-ubuntu, tags-math-sx, and threads-ask-ubuntu, we normalize the entire time ranges of all datasets to the same value. For the time window, we always set the time window as $T_w = 0.1T$.

\vspace{-2mm}
\subsection{Baselines and the experiment setup} \label{sec:base}

First, no previous baselines can be directly applied to temporal hypergraphs, let alone to predict the full spectrum of higher-order patterns. We properly revise their setting up to make them fit temporal hypergraphs and the settings of higher-order pattern prediction.

\textbf{Heuristic Methods}: Given a triplet of interest $(\{u,v\},w,t)$, we project all the hyperedges before $t$ into a static hypergraph. Then we transform the hypergraph into a traditional graph and calculate the corresponding features, including $3$-way Adamic Adar Index (\textbf{$3$-AA}), $3$-way Jaccard Coefficient (\textbf{$3$-JC}), $3$-way Preferential Attachment(\textbf{$3$-PA}), and the arithmetic mean of traditional \textbf{AA}, \textbf{JC} and \textbf{PA}.
We have not adopted the PageRank scores generalized for higher-order pattern prediction proposed by~\cite{nassar2020neighborhood} because computing those scores in temporal hypergraphs may take much time since the hypergraph structures change over time. 

Given a node triplet $(u,v,w)$, the $3-$way JC, AA, and PA, can be calculated as follows: $3-$AA: $f_{uvw} = \sum_{i \in N(u) \cap N(v) \cap N(w)} \frac{1}{\log |N(i)|}$;
$3-$JC: $f_{uvw} = \frac{|N(u) \cap N(v) \cap N(w)|}{|N(u) \cup N(v) \cup N(w)|}$;
$3-$PA: $f_{uvw} = |N(v)| \cdot |N(v)| \cdot |N(w)|$,
where $f_{uvw}$ is the corresponding feature, $N(i)$ means the set of neighbours of node $i$.

For the arithmetic mean of AA, PA, and JC, we calculate the three pairwise features in the traditional way, denote as $f_{uv}$, $f_{uw}$, and $f_{vw}$. Then we compute the arithmetic means of these features to make predictions. Indeed, the geometric and harmonic means have a similar results, so we only report the arithmetic means here. 

We impose a two-layer nonlinear network to first expand the dimension to $10$ and then do the prediction.

\textbf{NHP}~\cite{yadati2020nhp} with code is provided \hyperlink{https://drive.google.com/file/d/1pgSPvv6Y23X5cPiShCpF4bU8eqz_34YE/view?usp=sharing}{\textbf{here}}. NHP is designed for hyperedge prediction in static hypergraphs. 
NHP can also predict the unseen hyperedges when testing. NHP generalizes the GCN method to hypergraphs and proposes a novel negative sampling method for hyperedges. Since it cannot handle temporal information, we manually separate the temporal hypergraph into snapshots. We set the $10\%$ of the total time as a snapshot, so totally $10$ snapshots. Then we train each snapshot to get the node embeddings for the corresponding time period. 
We search the GCN hidden size in \{256, 512\} and the number of input node features in \{16, 32, 64\}. 

For \textbf{JODIE}, \textbf{TGAT}, and \textbf{TGN}, we first transform the hypergraph to a traditional graph by expanding all the hyperedges to a complete graph. The timestamp remains the same. Thus, the temporal hypergraph becomes a traditional temporal graph. We split the data according to the setting in Sec.~\ref{sec:exp}. We train the model to get the best validation score, and then output all the node embeddings and apply our decoders for different tasks.

\textbf{JODIE}~\cite{kumar2019predicting} with code is provided
\hyperlink{https://github.com/srijankr/jodie}{\textbf{here}}. JODIE is proposed to learn the embeddings of users and items in a temporal graph. Here, we adapt it to the temporal hypergraph problem. We train JODIE on the transferred graph. For each node, JODIE aims to minimize the distances between the predicted embedding and the embeddings of the connected nodes. JODIE processes edges in the temporal order and update the current node embeddings with temporal information. 
The model is trained for 50 epochs. For \textbf{tags-math-sx}, and \textbf{tags-ask-ubuntu}, we search the dynamic embedding dimension in \{64, 128\} and the best performance is reported. For \textbf{DAWN} and \textbf{congress-bills}, the dimension is in \{32, 64\}. For \textbf{threads-ask-ubuntu}, the dimension is in \{50,100\}. 

\textbf{TGAT}~\cite{xu2020inductive} with code is provided \hyperlink{https://openreview.net/forum?id=rJeW1yHYwH}{\textbf{here}}. It is designed  for link prediction tasks and node embedding tasks in the temporal network. TGAT generalizes the GAT mechanism to the temporal network by selecting both the topological neighbour and temporal neighbour in history. 
We train TGAT on the transferred graph to get node embeddings. We tune the parameter to get reasonable embeddings for all the nodes. We 
set the number of attention heads to 2, set the number of graph attention layers to 2, and use 100 as their default hidden dimension. We search the degree of their neighbor sampling in \{10, 20, 30\} and choose the best model. 

\textbf{TGN}~\cite{tgn_icml_grl2020} with code is provided \hyperlink{https://github.com/twitter-research/tgn}{\textbf{here}}. TGN is designed for temporal network problems. It introduces memory modules and combines them with graph attention or other embedding methods 
to substantially increase efficiency and performance. 
The memory module stores a vector for each node that appears in the model at time $t$. 
Each interaction among the nodes is a message. The memory module is updated according to the messages. 
We search the degree of neighbours to sample in \{10, 20, 30\} and the dimension of message in \{50, 100\} and choose the setting with the best AUC. 

\textbf{\proj} For our model, we search the TRW hyperparameter $\alpha$ in \{$1e-6$, $1e-5$, $1e-4$\}, the number of TRWs $M$ in \{64, 128\} and the step length $m$ in \{2,3\} and report the best result. Further experiments about parameter sensitivity are conducted and showed in Sec.~\ref{sup:para_sensitivity}. The experiments show that our model is robust when applying different parameters. 
For the encoder, we set the distance encoding feature dimension to $108$. $F_2$ is a one-layer nonlinear neural network, with output dimension $108$. $F_1$ is a two-layer nonlinear neural network and both hidden dimension and output dimension are $108$. For the time encoder in $F_3$, we set the dimension to $172$, \ie, $d = 172$. For the decoders, $F_4$ is a two-layer nonlinear neural network, with a hidden dimension $172$. 
Note that the dimension we choose is not crucial and our model is robust enough to a wide range of dimensions. 
$F_{w_p}$, $F_{\mu_p}$, and $F_{\sigma^2_p}$ are one-layer nonlinear neural networks. 

For all the models, we use Adam to optimize the neural network with learning rate $1e-4$. 
Since the data instances in different classes are extremely unbalanced (as illustrated in Table \ref{fig:data-stat} a)), we uniformly sample the same number of data instances from each class for training/evaluation/testing. The number of sampling instances is governed by the minimum size of the four classes. The sampled balanced dataset is shared across different models. All prediction performance is averaged across five experiments with the same random sampling process. 

\end{document}